\documentclass[12pt]{article}

\usepackage[letterpaper,margin=1in]{geometry}

\usepackage{fancyvrb}

\usepackage{amssymb,latexsym,amsmath,graphicx,amsthm,color,mathrsfs}

\usepackage{tikz}
\usetikzlibrary{trees,arrows,automata,shapes,snakes,decorations,decorations.pathmorphing}

\newenvironment{mytikz}[1][0em]{
\begin{tikzpicture}[>=latex,auto,node distance=6.75em,baseline={([yshift=#1]current bounding box.east)}]

  \tikzstyle{w}=[draw,circle,minimum size=3em]

  \tikzstyle{e}=[draw,minimum size=2.5em,node distance=3cm]
  
  \tikzstyle{d}=[densely dashed]

  \tikzstyle{every edge}=[draw,font=\footnotesize]
  
  \tikzstyle{every label}=[font=\footnotesize]

  \tikzstyle{zz}=[decorate,decoration={snake,post=lineto,post length=6pt}]
  
  \tikzstyle{ev}=[anchor=west,node distance=2em]

  \tikzstyle{bigger}=[minimum size=4em]

  \tikzstyle{bit-farther}=[node distance=2.3cm]

  \tikzstyle{farther}=[node distance=3cm]

  \tikzstyle{farthest}=[node distance=4cm]

  \tikzstyle{l}=[node distance=3.5em]
}{\end{tikzpicture}}

\newcommand{\imp}{\to}                   
\renewcommand{\iff}{\leftrightarrow}     
\newcommand{\from}{\leftarrow}           

\newcommand{\Prop}{\mathscr{P}}   
\newcommand{\Agnt}{\mathscr{A}}   
\newcommand{\Nat}{\mathbb{N}}     
\newcommand{\Actm}{\mathfrak{A}}   

\newcommand{\restricted}{\mathfrak{R}}  

\newcommand{\pre}{\mathsf{pre}} 

\newcommand{\DEL}{\mathsf{DEL}}   
\newcommand{\DETL}{\mathsf{DETL}} 
\newcommand{\YDEL}{\mathsf{YDEL}}  
\newcommand{\RDETL}{\mathsf{RDETL}} 
\newcommand{\lsetl}{{L_{\mathsf{SETL}}}}          
\newcommand{\ldetl}{{L_{\mathsf{DETL}}}}          
\newcommand{\lydel}{{L_{\mathsf{YDEL}}}}		
\newcommand{\lrdetl}{L_{\mathsf{RDETL}}}    

\newcommand{\leadsfrom}{%
  \mathrel{\raisebox{0em}{%
      \rotatebox[origin=c]{90}{%
        \reflectbox{\rotatebox[origin=c]{90}{$\leadsto$}}}}}}

\newcommand{\may}[1]{\langle{#1}\rangle} 

\newcommand{\pt}{{*}}              

\newcommand{\Sub}{\mathsf{Sub}}
\newcommand{\Cl}{\mathsf{Cl}}

\newcommand{\eqdef}{\stackrel{\mbox{\tiny def}}{=}}

\DeclareMathOperator{\mydepth}{d}                 

\theoremstyle{definition}
\newtheorem{definition}{Definition}[section]
\newtheorem{theorem}[definition]{Theorem}
\newtheorem{lemma}[definition]{Lemma}

\newtheorem{corollary}[definition]{Corollary}
\newtheorem{example}[definition]{Example}
\newtheorem{remark}[definition]{Remark}

\newcommand{\mytitle}{Logics of Temporal-Epistemic Actions}

\newcommand{\bryan}{Bryan Renne}

\newcommand{\bryanFunding}{Funded by an Innovational Research
  Incentives Scheme Veni grant from the Netherlands Organisation for
  Scientific Research (NWO).}

\newcommand{\joshua}{Joshua Sack}

\newcommand{\joshuaFunding}{Currently supported by the Netherlands
  Organisation for Scientific Research VIDI project 639.072.904.}

\newcommand{\audrey}{Audrey Yap}

\title{\mytitle{}}

\author{\bryan{}\footnote{\bryanFunding} \and
  \joshua{}\footnote{\joshuaFunding} \and \audrey{}}

\date{}

\begin{document}
\maketitle


\begin{abstract}
  We present \emph{Dynamic Epistemic Temporal Logic}, a framework for
  reasoning about operations on multi-agent Kripke models that contain
  a designated temporal relation. These operations are natural
  extensions of the well-known ``action models'' from Dynamic
  Epistemic Logic.  Our ``temporal action models'' may be used to
  define a number of informational actions that can modify the
  ``objective'' temporal structure of a model along with the agents'
  basic and higher-order knowledge and beliefs about this structure,
  including their beliefs about the time.  In essence, this approach
  provides one way to extend the domain of action model-style operations
  from atemporal Kripke models to temporal Kripke
  models in a manner that allows actions to control the flow of time.
  We present a number of examples to illustrate the subtleties
  involved in interpreting the effects of our extended action models
  on temporal Kripke models.  We also study preservation of important
  epistemic-temporal properties of temporal Kripke models under
  temporal action model-induced operations, provide complete
  axiomatizations for two theories of temporal action models, and
  connect our approach with previous work on time in Dynamic Epistemic
  Logic.
 \end{abstract}

\section{Introduction}

Anyone who has been late to an appointment or missed a deadline is
aware that it is often difficult to keep track of time. This basic
difficulty is the motivation for this paper, which presents a
framework called \emph{Dynamic Epistemic Temporal Logic} that allows
us to reason about epistemic agents' changing beliefs about time, from
one point in time to the next. We will develop this framework by
combining techniques from the traditions of \emph{Epistemic Temporal
  Logic} (ETL) \cite{ParRam03} and \emph{Dynamic Epistemic Logic}
(DEL) \cite{BalMos04,BalMosSol98,BalDitMos08,BenEijKoo06,DitHoeKoo07}.

Section \ref{section:system} introduces the syntax and semantics of
Dynamic Epistemic Temporal Logic (DETL). Section
\ref{section:examples} then highlights several features of this system
by presenting a number of examples that illustrate different ways of
measuring the time at a world in a model. The proof system and
completeness results for DETL appear in Section
\ref{section:proof-system}. In Section \ref{section:preservation}, we
study the preservation under updates of several model-theoretic
properties that one might wish to enforce so as to ensure models have
sensible temporal structure. Finally, we conclude in Section
\ref{section:connections} by connecting DETL to other work concerned
with adding time to DEL.


\section{Our System}
\label{section:system}

We begin with a nonempty finite set $\Agnt$ of \emph{agents} and a
disjoint nonempty set $\Prop$ of \emph{(propositional) letters}.  Our
semantics is based on \emph{Kripke models (with yesterday).}  These
are structures $M=(W^M,\to^M,\leadsfrom^M,V^M)$ consisting of a
nonempty set $W^M$ of \emph{(possible) worlds}, a binary epistemic
accessibility relation $\to^M_a$ for each $a\in\Agnt$ indicating the
worlds $w'\from^M_a w$ agent $a\in\Agnt$ considers possible at $w$, a
binary temporal accessibility relation $\leadsfrom^M$ indicating the
worlds $w' \leadsto^M w$ to be thought of as a ``yesterday'' of (i.e.,
fall one clock-tick before) world $w$,\footnote{For all structures
  $X$, let $\leadsto^X$ denote the converse of $\leadsfrom^X$ and let
  $\from^X_a$ denote the converse of $\to^X_a$. Our discussion of
  temporal issues will typically use $\leadsto$ rather than
  $\leadsfrom$ because the former follows the natural direction of
  time's flow.  Here and elsewhere, we will omit superscripts on
  relations when doing so ought not cause confusion.} and a
\emph{(propositional) valuation\/} $V^M:\Prop\to\wp(W^M)$ indicating
the set $V^M(p)$ of worlds at which propositional letter $p\in\Prop$
is true. For now, we do not place any restrictions on the behavior of
these relations, but later (in Definition
\ref{definition:model-properties}) we will introduce several desirable
properties that they will typically have in concrete examples. For a
binary relation $R$, a pair $(w,v)\in R$ is often called an ``$R$
arrow.'' A \emph{pointed Kripke model (with yesterday)}, sometimes
called a \emph{situation}, is a pair $(M,w)$ consisting of a Kripke
model and a world $w\in W^M$ called the \emph{point}.  To say that a
Kripke model (pointed or not) is \emph{atemporal} means that it
contains no $\leadsfrom$ arrows.

Pointed Kripke models $(M,w)$ describe fixed (i.e., ``static'')
epistemic-temporal situations in which agents have certain beliefs
about time, propositional truth, and the beliefs of other agents.  We
now define \emph{(epistemic-temporal) action models}, which transform
a situation $(M,w)$ into a new situation $(M[U],(w,s))$ according to a
certain ``product operation'' $M\mapsto M[U]$
defined in a moment (in Definition~\ref{definition:semantics}).

\begin{definition}[Action Models]
  \label{definition:update-frames}
  Let $F$ be a nonempty set of formulas.  An \emph{action model $U$
    over $F$} is a structure $(W^U,\to^U,\leadsfrom^U,\pre^U)$
  satisfying the following.
  \begin{itemize}
  \item $W^U$ is a nonempty finite set of informational \emph{events}
    the agents may experience.

  \item For each $a\in\Agnt$, the object $\to^U_a$ is a binary
    \emph{(epistemic) accessibility relation}. The relation $\to_a^U$
    designates the events $s'\from^U_a s$ that agent $a$ thinks are
    consistent with her experience of event $s$.

  \item $\leadsfrom^U$ is a binary \emph{temporal relation} indicating
    the events $s' \leadsto^U s$ that occur as a ``yesterday'' of
    (i.e., fall one time-step before) event $s$.

  \item $\pre^U:W^U\to F$ is a \emph{precondition function} assigning
    a \emph{precondition (formula)} $\pre^U(s)\in F$ to each event
    $s$.  The precondition $\pre^U(s)$ of event $s$ is the
    condition that must hold in order for event $s$ to occur.
  \end{itemize}
  A \emph{pointed action model over $F$}, sometimes called an
  \emph{action}, is a pair $(U,s)$ consisting of an action model $U$
  over $F$ and an event $s\in W^U$ called the \emph{point}.  To say
  that an action model (pointed or not) is \emph{atemporal} means that
  it contains no $\leadsfrom$ arrows. We define the following sets:
  \begin{itemize}
  \item $\Actm(F)$ is the set of action models over $F$,

  \item $\Actm^a(F)$ is the set of atemporal action models over $F$,

  \item $\Actm_*(F)$ is the set of pointed action models over $F$, and

  \item $\Actm^a_*(F)$ is the set of pointed atemporal action models
    over $F$.
  \end{itemize}
\end{definition}

Atemporal action models were developed by Baltag, Moss, and Solecki
\cite{BalMos04,BalMosSol98} and have been adapted or extended in
various ways in the Dynamic Epistemic Logic literature in order to
reason about knowledge and belief change; see the textbook
\cite{DitHoeKoo07} for details and references.  Our contribution here
is the inclusion of temporal arrows $\leadsfrom$ within action
models. To say more about this, we first introduce some additional
terminology.

\begin{definition}[Progressions, Histories, Depth $\mydepth(w)$]
  \label{definition:progression}\label{definition:run}\label{definition:depth}
  A \emph{progression} is a finite nonempty sequence
  $\may{w_i}_{i=0}^n$ of states having $w_i\leadsto w_{i+1}$ for each
  $i<n$.  We say that a progression $\may{w_i}_{i=0}^n$ \emph{begins
    at $w_0$ and ends at $w_n$}.  The \emph{length} of a progression
  $\may{w_i}_{i=0}^n$ is the number $n$, which is equal to the number
  of $\leadsto$ arrows it takes to link up the states making up the
  progression (i.e., one less than the number of states in the
  progression).  A \emph{past-extension} of a progression $\sigma$ is
  another progression obtained from $\sigma$ by adding zero or more
  extra states at the beginning of the sequence (i.e., in the
  ``past-looking direction'' from $x$ to $y$ in the arrow $x\leadsfrom
  y$).  A past-extension is \emph{proper} if more than zero states
  were added.  A \emph{history} is a progression that has no proper
  past-extension.  For each state $w$, we define $\mydepth(w)$ as
  follows: if there is a maximum $n\in\Nat$ such that there is a
  history of length $n$ that ends at $w$, then $\mydepth(w)$ is this
  maximum $n$; otherwise, if no such maximum $n\in\Nat$ exists, then
  $\mydepth(w) = \infty$.  We call $\mydepth(w)$ the \emph{depth of
    $w$}.  A state $w$ satisfying $\mydepth(w)=0$ is said to be
  \emph{initial}.
\end{definition}

We will present a number of examples shortly showing that the
inclusion of temporal $\leadsfrom$ arrows in both Kripke models and
action models allow us to reason about time in a Dynamic Epistemic
Logic-style framework.  The basic idea is this: if the depth
$\mydepth(w)$ of a world $w$ is finite, then the depth $\mydepth(w)$
of $w$ indicates the time at $w$; likewise, if the depth $\mydepth(s)$
of an event $s$ is finite, then the depth $\mydepth(s)$ of $s$
indicates the relative time at which event $s$ takes place.  To make
these notions coherent and useful, there are a number of things we
will do. First, we introduce our multi-modal language $\ldetl$ having
doxastic modalities $\Box_a\varphi$ (``agent $a$ believes $\varphi$'')
for each $a\in\Agnt$, the temporal modality $[Y]\varphi$ (``$\varphi$
was true `yesterday' (i.e., one time-step ago)''), and action model
modalities $[U,s]\varphi$ (``after action $(U,s)$ occurs, $\varphi$ is
true'').

\begin{definition}[Languages $\ldetl$ and $\lsetl$]
  \label{definition:language}
  The set $\ldetl$ of \emph{formulas $\varphi$ of Dynamic Epistemic
    Temporal Logic} and the set $\Actm_*(\ldetl)$ of pointed action
  models over $\ldetl$ are defined by the following recursion:
  \[
  \begin{array}{lcl}
    \varphi &::=&
    p \mid \lnot\varphi \mid \varphi\land\varphi \mid
    \Box_a\varphi \mid [Y]\varphi
    \mid [U,s]\varphi
    \\
    &&
    \text{\footnotesize
      $p\in\Prop$,
      $a\in\Agnt$,
      $(U,s)\in\Actm_\pt(\ldetl)$
    }
  \end{array}
  \]
  To say that a formula $\varphi$ is \emph{atemporal} means that every
  action model used in the formation of $\varphi$ according to the
  above recursion is atemporal.  We define the set $\lsetl$ of
  \emph{formulas of Simple Epistemic Temporal Logic} as the set of
  $\ldetl$-formulas that do not contain any action model modalities
  $[U,s]$.  We use the usual abbreviations from classical
  propositional logic to represent connectives other than those in the
  language, including those for the propositional constants $\top$
  (truth) and $\bot$ (falsehood); also,
  $\may{U,s}\eqdef\lnot[U,s]\lnot$, $\may{Y}\eqdef\lnot[Y]\lnot$, and
  $\diamondsuit_a\eqdef\lnot\Box_a\lnot$.
\end{definition}

\begin{definition}[Past State]
  \label{definition:past-state}
  Let $U$ be an action model.  A \emph{past state} is an event $s$ in
  $U$ that has no yesterday: there is no $s'\leadsto^U s$.
\end{definition}

Every history $s_0\leadsto s_1\leadsto s_2\leadsto\cdots\leadsto s_n$
begins with a past state (see
Definition~\ref{definition:progression}). This past state $s_0$ plays
a special role in the semantics by copying part or all of the initial
state of affairs before the remaining events in the history take
place.  The next definition shows how this is done.  The sequential
execution of successive events $s_1,\ldots,s_n$ then transform this
copy.

\begin{definition}[Semantics]
  \label{definition:semantics}
  We define the binary truth relation $\models$ between pointed Kripke
  models (written without delimiting parenthesis) and formulas by an
  induction on formula construction that has standard Boolean cases
  and the following non-Boolean cases.
  \begin{itemize}
  \item $M,w\models\Box_a\varphi$ means $M,v\models\varphi$ for each
    $v\from^M_aw$.

  \item $M,w\models[Y]\varphi$ means
    $M,v\models\varphi$ for each $v\leadsto^M w$.

  \item $M,w\models[U,s]\varphi$ means $M,w\models\pre^U(s)$ implies
    $M[U],(w,s)\models\varphi$, where
    \begin{itemize}
    \item $W^{M[U]} \eqdef \{ (v,t) \in W^M \times W^U \mid
      M,v\models\pre^U(t)\}$.

    \item We have $(v,t)\to^{M[U]}_a(v',t')$ if and only if both
      $v\to^M_av'$ and $t\to^U_a t'$.

    \item We have $(v',t') \leadsto^{M[U]} (v,t)$ if and only if we
      have at least one of the following:
      \begin{itemize}
      \item $v' \leadsto^M v$, $t'=t$, and $t$ is a past state; or

      \item $v'=v$ and $t' \leadsto^U t$.
      \end{itemize}

    \item $V^{M[U]}(p) \eqdef \{(v,t)\in W^{M[U]} \mid v \in V^M(p)\}$.
    \end{itemize}
  \end{itemize}
  Formula validity $\models\varphi$ means that $M,w\models\varphi$ for
  each pointed Kripke model $(M,w)$.  When it will not cause
  confusion, we will write the application of a function $f$ to a
  paired world $(v,t)\in W^{M[U]}$ as $f(v,t)$ instead of the more
  cumbersome $f((v,t))$. We may write $\models_\DETL$ in place of
  $\models$ later in the paper when other notions of truth are
  defined.
\end{definition}

After taking the update product $M\mapsto M[U]$, the epistemic
relation $\to_a$ behaves as it does in DEL \cite{DitHoeKoo07}; that
is, one pair is epistemically related to another iff they are related
componentwise. This is analogous to the notion of \emph{synchronous
  composition} in process algebra~\cite{Glabbeek_thelinear}. However,
our relations $\to_a$ are epistemic, whereas the relations in process
algebra are temporal.

The behavior of our temporal relation $\leadsto$ after the update
product $M\mapsto M[U]$ is analogous to the notion of \emph{asynchronous
  composition} from process
algebra~\cite{Aceto99structuraloperational}: one component of the pair
makes the transition while the other component remains fixed. However,
in our case, if $(v',t') \leadsto^{M[U]} (v,t)$, then the component it
is that makes the transition depends on whether $t$ is a past state.
If $t$ is indeed a past state, then the first component makes the
transition ($v'\leadsto v$) and the second component remains fixed
($t'=t$).  Otherwise, if $t$ is not a past state, then it is instead
the first component that remains fixed ($v'=v$) and the second
component that makes the transition ($t'\leadsto t$).

\begin{definition}[Epistemic Past State]
  \label{definition:epistemic-past-state}
  Let $U$ be an action model. An \emph{epistemic past state} is a past
  state $s$ in $U$ whose precondition is a validity (i.e.,
  $\models\pre^U(s)$) and whose only epistemic arrows are the
  reflexive arrows $s\to^U_a s$ for each agent $a\in\Agnt$.
\end{definition}

Like a past state (Definition~\ref{definition:past-state}), an
epistemic past state $s_0$ in a history $s_0\leadsto\cdots\leadsto
s_n$ plays the special role of copying the initial state of affairs
before the remaining events in the history take place.  However, there
is a key difference: a past state may copy only part of the initial
state of affairs, whereas an epistemic past state will always make a
complete copy.  A later result (Theorem~\ref{theorem:past-state}) will
explain this further. Therefore, a history $s_0\leadsto\cdots\leadsto
s_n$ beginning with an epistemic past state $s_0$ may be thought of as
describing the following construction: the epistemic past state $s_0$
makes a complete copy of the initial state of affairs (thereby
remembering the ``past'' just as it was) and then the remaining events
$s_1,\dots,s_n$ transform this copy (appending ``future'' states of
affairs one by one).  A series of examples in
Section~\ref{section:examples} will explain this in further detail.

In a significant departure from the temporal logic literature, our
language does not include a future operator $[T]$ (for ``tomorrow'')
that acts as a converse of our yesterday operator $[Y]$:
\begin{equation}
  M,w\models [T]\varphi \text{ means }
  M,v\models\varphi \text{ for each } v\leadsfrom^M w.
  \label{eq:T}
\end{equation}
The reason we omit the temporal operator $[T]$ is that the update
modal $[U,s]$ already functions as a forward-looking temporal operator
of a different sort. Such operators $[U,s]$ are parametrized future
operators.  A common approach (see \cite{Bal08}) to relating
parametrized and unparametrized operators is to have the
unparametrized operator quantify over the parametrized operators.  A
tomorrow operator defined that way would be significantly different
from $[T]$ as defined in \eqref{eq:T}.  Furthermore, having the $[T]$
as well as the update operators $[U,s]$ results in unintuitive
updates, which will be explained from various perspectives over the
course of the paper. For now, it suffices to say that our framework
has a static past (via the operator $[Y]$) and a dynamic future (via
the Kripke model-changing operators $[U,s]$).  To help make sense of
these notions of ``past'' and ``future,'' we will impose in all
concrete examples a number of restrictions on Kripke models and on
action models that allow us to provide a coherent and meaningful
account of time and of the flow of time within the setting of our
framework.  Many of the restrictions on Kripke models can be found
elsewhere in the literature of DEL and ETL
\cite{BenGerHos09,BenGerPac07,deglowwit11,Sac10,Sac08,Yap07}; however,
the identification of relevant action model-specific restrictions and
the use of Kripke model restrictions in action models having
$\leadsto$ arrows is, to the best of our knowledge, new (though it
builds off the authors' early work in \cite{RenSacYap09}).

\begin{definition}
  \label{definition:model-properties}
  The following properties may apply to Kripke models or to (pointed)
  action models.
  \begin{itemize}
  \item \emph{Persistence of Facts\/} (for Kripke models only):
    $w\leadsto w' \Rightarrow(w\in V(p)\Leftrightarrow w'\in V(p))$
    for all $w$ and $w'$.

    This says that propositional letters retain their values across
    temporal $\leadsto$ arrows.

  \item \emph{Depth-Definedness\/}: $\mydepth(w)\neq\infty$ for all
    $w$.

    This says that every world/event has a finite depth.

  \item \emph{Knowledge of the Past\/}: $(w'\leadsto w\to_a
    v)\Rightarrow\exists v'(v'\leadsto v)$ for all $a\in\Agnt$, $w'$,
    $w$, and $v$.

    This says that agents know if there is a past (i.e., a state
    reachable via a backward step along a $\leadsto$ arrow).

  \item \emph{Knowledge of the Initial Time\/}: $w\to_av\land
    \lnot\exists w'(w'\leadsto w)\Rightarrow \lnot\exists
    v'(v'\leadsto v)$ for all $a\in\Agnt$, $w$, and $v$.

    This says that agents know if there is no past.

  \item \emph{Uniqueness of the Past\/}: $(w'\leadsto w\land w''\leadsto
    w)\Rightarrow (w'=w'')$ for all $w'$, $w$, and $w''$.

    This says that there is only one possible past.

  \item \emph{Perfect Recall\/}: $(w\leadsto v\to_a v')\Rightarrow
    \exists w'(w\to_aw'\leadsto v')$ for all $a\in\Agnt$, $w$, $v$,
    and $v'$.

    This says that agents do not forget what they knew in the past.

  \item \emph{Synchronicity\/}: The structure is depth-defined and
    $w\to_a w'$ implies $\mydepth(w)=\mydepth(w')$ for all $a\in\Agnt$,
    $w$, and $w'$.

    This says that there is no uncertainty, disagreement, or
    mistakenness among the agents with regard to the depth of a
    world/event.

  \item \emph{History Preservation\/} (for action models $U$ only):
    $s'\leadsto^U s$ implies $\models\pre^U(s)\imp\pre^U(s')$ for all
    $s'$ and $s$; further, every past state in $U$ is an epistemic
    past state.

    This says that a non-initial event $s$ can take place only if its
    predecessor $s'$ can as well, and initial events can always take
    place (Definition~\ref{definition:epistemic-past-state}).

  \item \emph{Past Preservation\/} (for pointed action models $(U,s)$
    only): The action model $U$ is history preserving; further, every
    progression $s_n\leadsto^U\cdots\leadsto^Us_0=s$ that ends at $s$
    can be past-extended to a history
    \[
    s_{n+m}\leadsto^U 
    \cdots \leadsto^U s_{n+1} \leadsto^U
    s_n\leadsto^U\cdots\leadsto^Us_0=s
    \]
    that begins at an epistemic past state $s_{n+m}$. To say an event
    $t\in W^U$ is \emph{past preserving} means the action $(U,t)$ with
    point $t$ is past preserving.

    This says that there is a ``link to the past'' (i.e., an epistemic
    past state) via a sequence of $\leadsfrom$ arrows.  (The
    forthcoming Theorem~\ref{theorem:past-state} describes the
    consequences of this property.)

  \item \emph{Time-advancing\/} (for pointed action models $(U,s)$
    only): The action $(U,s)$ is past preserving and the point $s$ is
    not a past state.

    This says that the ``past'' is at least one time-step away.
  \end{itemize}
\end{definition}
For the moment, we do not require that our Kripke models or action
models satisfy any of the above properties. This will change in
Section~\ref{section:preservation}, where we study the preservation of
Kripke model properties under action models satisfying appropriate
properties, and in Section~\ref{section:connections}, where we impose
a number of these properties in order to study connections between our
framework and other approaches to the study of time in Dynamic
Epistemic Logic.

A note on the depth of worlds (Definition~\ref{definition:run}) in the
updated model $M[U]$: if world $w$ and state $t$ both have finite
depth, then the maximum depth of the resulting world $(w,t)\in
W^{M[U]}$ is $\mydepth(w)+\mydepth(t)$.  The reason: we can take at
most $\mydepth(t)$ backward steps in the second coordinate (fix the
first coordinate and proceed backward in the second until a past state
is reached), and we can take at most $\mydepth(w)$ backward steps in
the first coordinate (fix the second coordinate and proceed backward
in the first until an initial world is reached).  The actual depth of
$(w,t)$ need not obtain its maximum: when stepping backward in either
coordinate (with the other fixed), we may reach a pair $(w',s')$ whose
world $w'$ violates the precondition of the state $s'$ (i.e.,
$M,w'\not\models\pre^U(s')$) and therefore the pair $(w',s')$ will not
be a member of $W^{M[U]}$. However, if the action model $U$ is history
preserving, then this problem is avoided and hence
$\mydepth(w,t)=\mydepth(w)+\mydepth(t)$.

Finally, we note that we can express finite depth explicitly in our
language.

\begin{theorem}
  \label{theorem:time}
  For each non-negative integer $n$, define the formulas
  \[
  D_n \eqdef \may{Y}^n[Y]\bot \land [Y]^{n+1}\bot \quad\text{and}\quad
  D'_n \eqdef \may{Y}^n[Y]\bot \enspace.
  \]
  We have $M,w\models D_n$ if and only if $\mydepth(w)=n$.  Further,
  if $M$ satisfies uniqueness of the past, then we have $M,w\models
  D'_n$ if and only if $\mydepth(w)=n$.
\end{theorem}

In the next section, we will discuss how the depth of a world $w$ can
be used as an explicit measure of the time at world $w$.  If one
adopts this measure of time, then it follows from
Theorem~\ref{theorem:time} that we can express the time of a world
within our language $\ldetl$: we say that ``the time at world $w$ is
$n$'' to mean that $w$ satisfies $D_n$.


\section{Examples}
\label{section:examples}

In this section, we will illustrate several features of our
system. This will demonstrate the way in which our system can
represent interesting epistemic situations as well as shed some light
on the ways in which time is treated in dynamic systems. First, we
will define \emph{explicit} and \emph{implicit} measures of time. An
\emph{explicit} measure of time is one in which the time of a world
$w$ in model $M$ can be determined solely by inspection of $M$.  As we
have seen, the depth of a world in a Kripke model (see Definition
\ref{definition:run}) can provide an explicit measure of the time at
that world, and this can be expressed explicitly in our language
(Theorem~\ref{theorem:time}).  In contrast, an \emph{implicit} measure of
time is one in which the time of a world $w$ in model $M$
\emph{cannot} be determined solely be inspection of $M$. For example,
if we measure the time at a world in $M$ by the number of updates that
have led up to $M$, this can provide an implicit measure of the time
at that world. These are not the only possible implicit and explicit
measures, but they are certainly natural ones within our $\DETL$
framework.

Now we will consider ways in which explicit and implicit
representations of time might differ. More specifically, we will
consider cases in which there is only a single update (implicitly
increasing the time by $1$) but at which the explicit time at the actual
world changes by a number other than $1$. Second, we will add in the
epistemic aspect, by demonstrating the ways in which agents can hold
differing beliefs about temporal and epistemic features of their
situation.

In all of these cases, we will begin with a basic Kripke model $(M,w)$
pictured in Figure~\ref{figure:M} in which agents $a$ and $b$ are
uncertain of the truth values of $p$ and $q$. The double outline
around world $w$ indicates that $w$ is the ``actual world'' (i.e.,
that $w$ is the point) and hence that $p$ and $q$ are actually
true. Throughout this section, we will assume that our epistemic
accessibility relations $\to_a$ and $\to_b$ are $\mathsf{S5}$ (i.e.,
that they are closed under reflexivity, transitivity, and symmetry);
however, for the sake of readability, we will not draw all transitive
arrows.

\begin{figure}
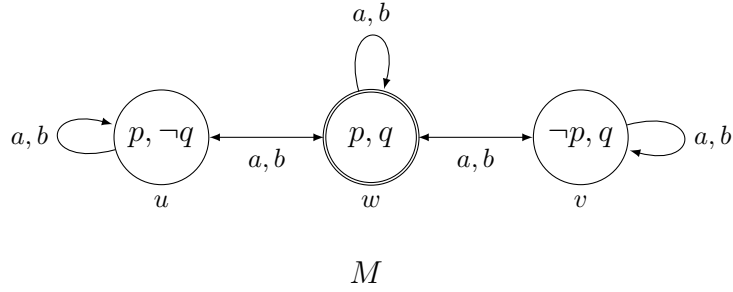

\begin{center}
  \begin{tabular}{c}
    \begin{mytikz}
      
      \node[w,double,label={below:$w$}] (w) {$p, q$};
      
      \node[w,right of=w,label={below:$v$}] (v) {$\lnot p, q$};
      
      \node[w,left of=w,label={below:$u$}] (u) {$p, \lnot q$};
      
      \path (w) edge[loop above] node{$a,b$} ();
      
      \path (w) edge[<->,swap] node{$a,b$} (v);
      
      \path (w) edge[<->] node{$a,b$} (u);
      
      \path (v) edge[loop right] node{$a,b$} ();
      
      \path (u) edge[loop left] node{$a,b$} ();
    \end{mytikz}
    \\\\
    $M$
  \end{tabular}
\end{center}
\caption{$(M,w)$, a situation in which $p, q\in\Prop$ are both true
  but the agents do not know that this is so. Agent arrows $\to_x$ are
  here implicitly closed under transitivity.}
\label{figure:M}
\end{figure}

\subsection{Explicit and Implicit Representations of Time}

We now discuss three examples related to discrepancies between the
number of updates that have occurred and the explicit time at a world
as determined by its depth. First, we will see ``standard'' behavior,
in which a single update increases the time of the actual world by $1$.

\begin{example}
\label{example:atemporal-update}
The action $(U_{\ref{figure:atemporal-update}},s)$
(Figure~\ref{figure:atemporal-update}) represents the public
announcement of $p$.  Applied to our initial situation $(M,w)$
(Figure~\ref{figure:M}), the result is
$(M[U_{\ref{figure:atemporal-update}}],(w,s))$.

\begin{figure}
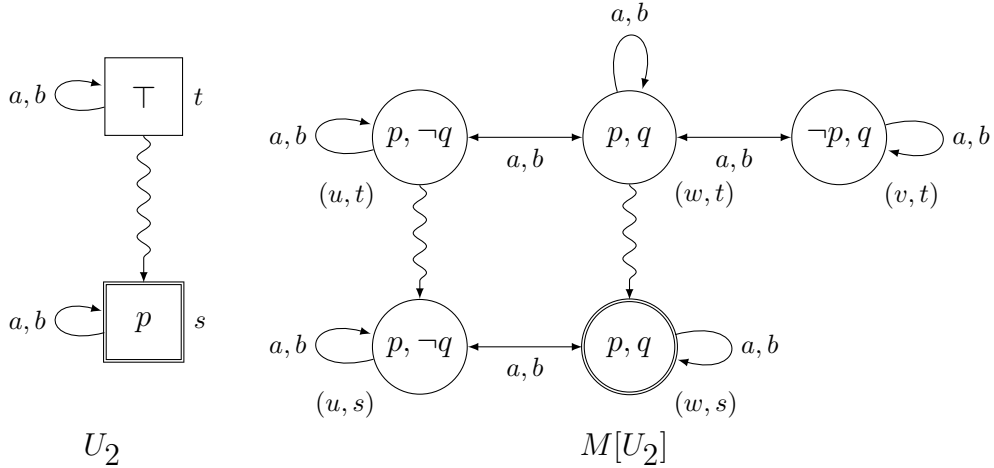

  \begin{center}
    \begin{tabular}{cc}
      \begin{mytikz}
        \node[e,label={right:$t$}] (t) {$\top$};
        
        \node[e,double,below of=t,label={right:$s$}] (s) {$p$};
        
        \path (t) edge[loop left] node{$a,b$} ();
        
        \path (t) edge[zz,->] node{} (s);
        
        \path (s) edge[loop left] node{$a,b$} ();
      \end{mytikz}
      &
      \begin{mytikz}
        \node[w,,xshift=2em,label={below left:$(u,t)$}] (ut) {$p, \lnot q$};
        
        \node[w,right of=ut,label={below right:$(w,t)$}] (wt) {$p, q$};
        
        \node[w,right of=wt,label={below right:$(v,t)$}] (vt) {$\lnot p, q$};
        
        \path (wt) edge[loop above] node{$a,b$} ();
        
        \path (wt) edge[<->,swap] node{$a,b$} (vt);
        
        \path (wt) edge[<->] node{$a,b$} (ut);
        
        \path (vt) edge[loop right] node{$a,b$} ();
        
        \path (ut) edge[loop left] node{$a,b$} ();
        
        \node[w,double,below of=wt,label={below right:$(w,s)$}] (ws) {$p, q$};
        
        \node[w,below of=ut,label={below left:$(u,s)$}] (us) {$p, \lnot q$};
        
        \path (ws) edge[loop right] node{$a,b$} ();
        
        \path (us) edge[loop left] node{$a,b$} ();
        
        \path (ws) edge[<->] node{$a,b$} (us);
        
        \path (wt) edge[zz,->] node{} (ws);
        
        \path (ut) edge[zz,->] node{} (us);
      \end{mytikz}
      \\
      $U_{\ref{figure:atemporal-update}}$ & 
      $M[U_{\ref{figure:atemporal-update}}]$
    \end{tabular}
  \end{center}
  \caption{Left: $(U_{\ref{figure:atemporal-update}},s)$, the public
    announcement of $p$. Right: The resulting situation
    $(M[U_{\ref{figure:atemporal-update}}], (w,s))$. Agent arrows
    $\to_x$ are here implicitly closed under transitivity. (Example
    \ref{example:atemporal-update})}
  \label{figure:atemporal-update}
\end{figure}

In the situation $(M,w)$ before the announcement, neither $a$ nor $b$
knew $p$; that is, $M,w\models \neg\Box_a p \wedge \neg \Box_b p$.  In
the situation $(M[U_{\ref{figure:atemporal-update}}],(w,s))$ after the
announcement, both know $p$; that is,
$M[U_{\ref{figure:atemporal-update}}],(w,s)\models \Box_a p \wedge
\Box_b p\enspace$.  But note that in the ``after'' model
$M[U_{\ref{figure:atemporal-update}}]$, we have a ``copy'' of the
original model $M$ consisting of the worlds $(u,t)$, $(w,t)$, and
$(v,t)$ and the arrows interconnecting these worlds.  As a result, we
can describe the agents' knowledge ``before'' and ``after'' all
together in the resultant situation:
\[
M[U_{\ref{figure:atemporal-update}}],(w,s)\models (\Box_a p \wedge
\Box_b p)\land \langle Y\rangle (\neg\Box_a p \wedge \neg \Box_b
p)\enspace.
\]
In words, ``$a$ and $b$ know $p$ but yesterday they did not.''
Further, we note that the (explicit) time at both the initial world
$w$ and its copy $(w,t)$ is $0$, and the time at the final world
$(w,s)$ is $1$.

The language of ordinary DEL is the atemporal fragment of $\ldetl$
without the $[Y]$ modality.  In this language, the only way to refer
to the agents' knowledge before the announcement is with reference to
the original situation $(M,w)$.  This is because ordinary DEL lacks
$\leadsto$ arrows, both in action models and in the underlying Kripke
models.
\end{example}

Example~\ref{example:atemporal-update} showed ``standard'' temporal
behavior: a single update increases the time of the actual world by
$1$.  Two obvious ways in which updates could be ``non-standard'' are
by not increasing the time when an update takes place or by increasing
the time by a number greater than $1$.

\begin{example}
  \label{example:no-time-increase}

  Here we consider the effect of the action
  $(U_{\ref{figure:no-time-increase}},t)$
  (Figure~\ref{figure:no-time-increase}) on our initial situation
  $(M,w)$ (Figure~\ref{figure:M}).  This action has a structure that
  is nearly identical to that of action
  $(U_{\ref{figure:atemporal-update}},s)$ in
  Example~\ref{example:atemporal-update}
  (Figure~\ref{figure:atemporal-update}); in fact, these actions are
  based on the same underlying action model (i.e.,
  $U_{\ref{figure:atemporal-update}}=U_{\ref{figure:no-time-increase}}$).
  However, the actual events of
  $(U_{\ref{figure:atemporal-update}},s)$ and
  $(U_{\ref{figure:no-time-increase}},t)$ are different.  The result
  of the update with the latter action is the pointed model
  $(M[U_{\ref{figure:no-time-increase}}],(w,t))$.  Note that the
  resultant actual world $(w,t)$ is among the worlds $(u,t)$, $(w,t)$,
  and $(v,t)$ that make up the ``copy'' of the initial model $M$.

  We have designed our system so that the initial world $w$ and its
  copy $(w,t)$ satisfy the same formulas.  In this way, we may
  identify each initial world with its copy, so that the collection of
  copied worlds (and the arrows interconnecting them) may be
  identified with the initial model itself.  This allows us to reason
  about what was the case in the initial model by evaluating formulas
  only within the resultant model.  In effect, we can ``forget'' the
  initial model because all of its information is copied over to the
  resultant model.

  To make this work, both $w$ and its copy $(w,t)$ must satisfy the
  same formulas.  We guarantee this by designing our system so that it
  ignores all ``future'' worlds, by which we mean the worlds
  accessible from the point only via a link $x\leadsto y$ from a
  ``past'' world $x$ to a ``future'' world $y$.\footnote{Formally, a
    $\ldetl$-formula $\varphi$ is true at a world $x$ if and only if
    $\varphi$ is true at $x$ even after we delete all worlds $y$
    satisfying the property that every path from $x$ to $y$ contains
    at least one $\leadsto$ arrow (followed in the ``forward''
    direction $z\leadsto z'$ from ``past'' $z$ to ``future'' $z'$).
    This is so because $\ldetl$ has no $[T]$ operator, as defined in
    \eqref{eq:T}.} So from the point of view of our theory, the
  time-$1$ worlds $(u,s)$ and $(w,s)$ in
  Figure~\ref{figure:no-time-increase} are ignored in the resultant
  time-$0$ situation $(M[U_{\ref{figure:no-time-increase}}],(w,t))$
  because the time-$1$ worlds can only be reached via a forward
  $\leadsto$ arrow.  This leaves only the ``copy'' of the initial
  model $M$ consisting of the worlds $(u,t)$, $(w,t)$, and $(v,t)$.
  The resultant situation
  $(M[U_{\ref{figure:no-time-increase}}],(w,t))$ is therefore
  equivalent to the initial situation $(M,w)$ from the point of view
  of our theory.  In other words, the action
  $(U_{\ref{figure:no-time-increase}},t)$ does not change the state of
  affairs at all.\footnote{This is similar to the way in which the
    ``do nothing'' Propositional Dynamic Logic (PDL) program
    $\mathit{skip}$ does not change the state of the system.  However,
    there is a difference: the PDL program does not change the
    structure of the model, though the action
    $(U_{\ref{figure:no-time-increase}},t)$ does.  Nevertheless, from
    the point of view of language equivalence, this change is
    inconsequential: the ``before'' situation and the ``after''
    situation satisfy the same $\ldetl$-formulas, and so our intention
    is that these situations are to be identified.}

  \begin{figure}
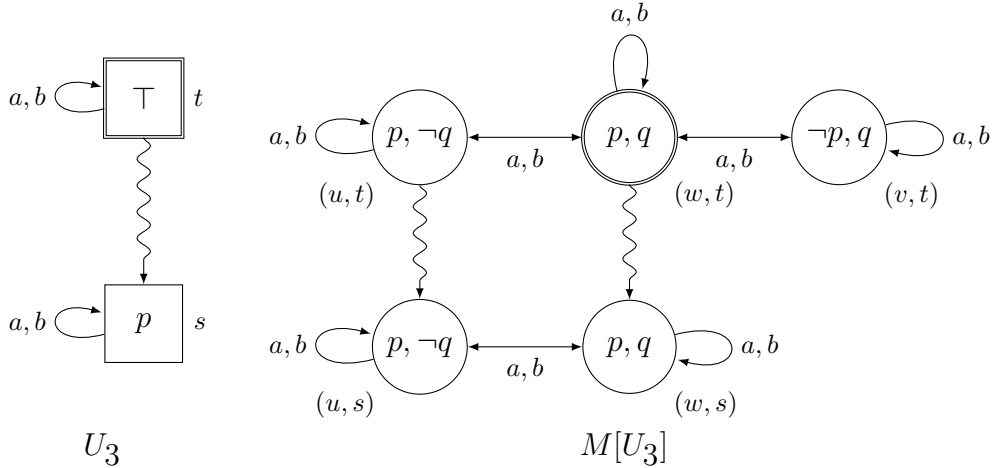

    \begin{center}
      \begin{tabular}{cc}
        \begin{mytikz}
          \node[e,double,label={right:$t$}] (t) {$\top$};
          
          \node[e,below of=t,label={right:$s$}] (s) {$p$};
          
          \path (t) edge[loop left] node{$a,b$} ();
          
          \path (t) edge[->,zz] node{} (s);
          
          \path (s) edge[loop left] node{$a,b$} ();
        \end{mytikz}
        &
        \begin{mytikz}
          \node[w,xshift=2em,label={below left:$(u,t)$}]
          (ut) {$p, \lnot q$};
          
          \node[w,double,right of=ut,label={below right:$(w,t)$}] (wt) {$p, q$};
          
          \node[w,right of=wt,label={below right:$(v,t)$}] (vt) {$\lnot p, q$};

          \path (wt) edge[loop above] node{$a,b$} ();
          
          \path (wt) edge[<->,swap] node{$a,b$} (vt);
          
          \path (wt) edge[<->] node{$a,b$} (ut);
          
          \path (vt) edge[loop right] node{$a,b$} ();
          
          \path (ut) edge[loop left] node{$a,b$} ();
          
          \node[w,below of=wt,label={below right:$(w,s)$}] (ws) {$p, q$};
          
          \node[w,below of=ut,label={below left:$(u,s)$}] (us) {$p, \lnot q$};
          
          \path (ws) edge[loop right] node{$a,b$} ();
          
          \path (us) edge[loop left] node{$a,b$} ();
          
          \path (ws) edge[<->] node{$a,b$} (us);
          
          \path (wt) edge[zz,->] node{} (ws);
          
          \path (ut) edge[zz,->] node{} (us);
          
        \end{mytikz}
        \\
        $U_{\ref{figure:no-time-increase}}$ &
        $M[U_{\ref{figure:no-time-increase}}]$
      \end{tabular}
    \end{center}
    \caption{Left: $(U_{\ref{figure:no-time-increase}},t)$, the future
      public announcement of $p$. Right: The resulting situation
      $(M[U_{\ref{figure:no-time-increase}}], (w,t))$. Agent arrows
      $\to_x$ are here implicitly closed under transitivity. (Example
      \ref{example:no-time-increase})}
    \label{figure:no-time-increase}
\end{figure}

\end{example}

\begin{remark}\label{remark:consequences-of-T}
  The previous two examples illustrate some motivation behind our
  choice not to include a $[T]$ operator, as defined in
  \eqref{eq:T}. If our language had included such an operator, two
  worlds that are intended to represent the exact same state of
  affairs could disagree about the truth of formulas. In Example
  \ref{example:no-time-increase}, a non-time-advancing update
  transforms $(M,w)$ into
  $(M[U_{\ref{figure:no-time-increase}}],(w,t))$, but the worlds $w$
  and $(w,t)$ are meant to represent the same situation and so should
  satisfy the same formulas.  However, $(M,w)$ and
  $(M[U_{\ref{figure:no-time-increase}}],(w,t))$ disagree on the truth
  of the formula $\langle T\rangle p$.  But as we have defined
  $\ldetl$ without the $[T]$-operator, we can easily show that for any
  $\ldetl$-formula $\varphi$, we have $M, w \models \varphi$ iff
  $M[U], (w,t) \models \varphi$ because $t$ is an epistemic past state
  (see Theorem~\ref{theorem:past-state}).  This avoids the problem
  illustrated here where two worlds that are supposed to represent the
  same situation disagree on the truth of formulas.

  Example \ref{example:atemporal-update} also illustrates the semantic
  difference between an action model operator $[U,s]$ and the operator
  $[T]$. The truth of $[U,s]\varphi$ is determined by evaluating
  $\varphi$ in a new model, while the truth of $[T]\varphi$ is
  determined by evaluating $\varphi$ within the model as it currently
  stands. In considering our initial model $(M,w)$ (Figure
  \ref{figure:M}), note that $M, w \not \models \may{T}p$. However, as
  we can see, $M, w \models
  \may{U_{\ref{figure:atemporal-update}},s}p$. So while we informally
  read the formula $\may{T}\varphi$ as claiming that $\varphi$ will
  hold ``tomorrow'' (and that there is at least one possible
  ``tomorrow'' world), this is only from the perspective of a static
  model---it does not consider all possible ways in which that model
  might evolve given different updates.

  The $[U,s]$ and $[T]$ operators are not the only options for
  ``tomorrow'' operators. Section \ref{section:system} mentions the
  way in which the update operators $[U,s]$ serve as dynamic
  parametrized operators. And it is also possible to define dynamic
  unparametrized operators that work by quantifying over the
  parametrized ones, and these also are different from $[T]$
  operators. Call our unparametrized operator $[N]$, and we can
  define: $M,w\models [N]\varphi$ if and only if $M,w\models
  [U,s]\varphi$ for all action models $(U,s)$ with only epistemic
  preconditions (this constraint is imposed on the action models, so
  as to ensure the well-foundedness of the semantics)
  \cite{Bal08}. Since all public announcement action models fall into
  this category, we have that $M,w\models \langle N\rangle p$ (since
  $M,w\models \langle U_{\ref{figure:atemporal-update}},s\rangle p$)
  and yet $M,w\not\models \langle T\rangle p$.

  We believe that these interpretive issues involving $[T]$ reflect
  its complex relationship with the update modalities. Indeed, $[T]$
  may not even have a clear interpretation in the context of our
  framework, which is part of the rationale for leaving it out of our
  language. But such a situation is not unusual in the epistemic logic
  tradition. For instance, a common system for modeling agents'
  beliefs is $\mathsf{KD45}$, whose epistemic accessibility relation
  does not have to be symmetric. In such systems, it is not clear that
  the converse of the accessibility relation has a clear semantic
  interpretation, but this is not viewed as problematic. So for us,
  the relation $\leadsfrom$ is one example of many from modal logic of
  a relation that has a corresponding modality but whose converse does
  not. (Also, the semantic asymmetry of having a $[Y]$ operator but no
  corresponding $[T]$ operator is analogous to an asymmetry in
  ordinary DEL, which has $[U,s]$ modalities but no converse
  $[U,s]^{-1}$ modalities.) As a result, our system $\DETL$ is one that
  has a dynamic parametrized future (accessed via the update modals
  $[U,s]$) and a static unparametrized past (accessed via the
  yesterday modal $[Y]$).
\end{remark}

\begin{example}
  \label{example:2-time-increase}
  In this example, the time at the actual world increases by $2$ even
  though only a single update $(U_{\ref{figure:2-time-increase}},r)$
  (Figure \ref{figure:2-time-increase}) takes place. With simple
  modifications of $(U_{\ref{figure:2-time-increase}},r)$, we could
  increase the time by any finite number.

  \begin{figure}
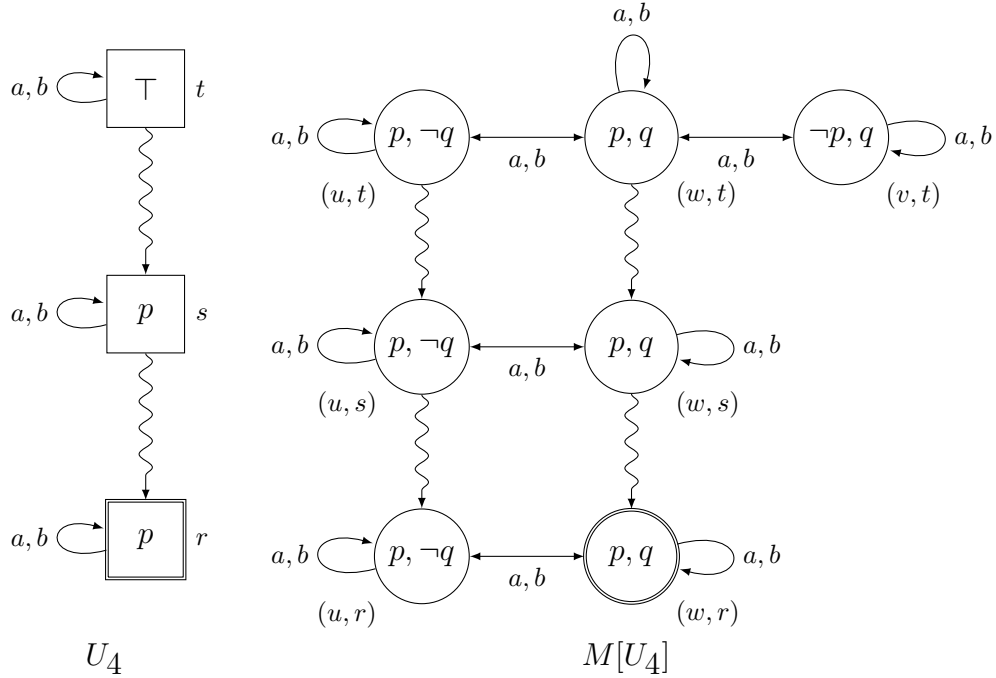

    \begin{center}
      \begin{tabular}{cc}
        \begin{mytikz}
          \node[e,label={right:$t$}] (t) {$\top$};
          
          \node[e,below of=t,label={right:$s$}] (s) {$p$};
          
          \node[e,double,below of=s,label={right:$r$}] (r) {$p$};
          
          \path (t) edge[loop left] node{$a,b$} ();
          
          \path (t) edge[->,zz] node{} (s);
          
          \path (s) edge[zz,->] node{} (r);
          
          \path (s) edge[loop left] node{$a,b$} ();
          
          \path (r) edge[loop left] node{$a,b$} ();
        \end{mytikz}
        &
        \begin{mytikz}
          
          \node[w,xshift=2em,label={below left:$(u,t)$}] (ut) {$p, \lnot q$};
          
          \node[w,right of=ut,label={below right:$(w,t)$}] (wt) {$p, q$};
          
          \node[w,right of=wt,label={below right:$(v,t)$}] (vt) {$\lnot p, q$};

          \path (wt) edge[loop above] node{$a,b$} ();
          
          \path (wt) edge[<->,swap] node{$a,b$} (vt);
          
          \path (wt) edge[<->] node{$a,b$} (ut);
          
          \path (vt) edge[loop right] node{$a,b$} ();
          
          \path (ut) edge[loop left] node{$a,b$} ();
          
          \node[w,below of=wt,label={below right:$(w,s)$}] (ws) {$p, q$};
          
          \node[w,below of=ut,label={below left:$(u,s)$}] (us) {$p, \lnot q$};
          
          \path (ws) edge[loop right] node{$a,b$} ();
          
          \path (us) edge[loop left] node{$a,b$} ();
          
          \path (ws) edge[<->] node{$a,b$} (us);
          
          \path (wt) edge[zz,->] node{} (ws);
          
          \path (ut) edge[zz,->] node{} (us);
          
          \node[w,double,below of=ws,label={below right:$(w,r)$}] (wr) {$p, q$};
          
          \node[w,below of=us,label={below left:$(u,r)$}] (ur) {$p, \lnot q$};
          
          \path (wr) edge[loop right] node{$a,b$} ();
          
          \path (ur) edge[loop left] node{$a,b$} ();
          
          \path (wr) edge[<->] node{$a,b$} (ur);
          
          \path (ws) edge[->,zz] node{} (wr);
          
          \path (us) edge[->,zz] node{} (ur);
          
        \end{mytikz}
        \\
        $U_{\ref{figure:2-time-increase}}$ &
        $M[U_{\ref{figure:2-time-increase}}]$
      \end{tabular}
    \end{center}
    \caption{Left: $(U_{\ref{figure:2-time-increase}},r)$, the double
      public announcement of $p$. Right: the resulting situation
      $(M[U_{\ref{figure:2-time-increase}}], (w,r))$. Agent arrows $\to_x$
      are here implicitly closed under transitivity. (Example
      \ref{example:2-time-increase})}
    \label{figure:2-time-increase}
  \end{figure}
\end{example}

These three examples demonstrate the differences between explicit and
implicit measures of time; in particular, the number of updates (the
implicit time) need not equal the depth of the actual world (the
explicit time). In this paper, we will adopt the convention that the
time at the actual world is measured by its depth (explicit
time). While this is by no means a necessary choice, it has the
advantage of allowing us to determine the time at a world solely by
inspection of the model to which it belongs.

\subsection{Agents Mistaken About Time}

Given that we are measuring the time at a world by its depth, we can
represent situations in which agents are unable to distinguish between
worlds that have different times. These situations can be brought
about by $\DETL$ actions, as the following example illustrates.

\begin{example}
\label{example:2-time-increase-multi}
In this example, $(U_{\ref{figure:2-time-increase-multi}},r)$
(Figure~\ref{figure:2-time-increase-multi}) represents the sequenced
public announcement of $p$ followed by the asynchronous semi-private
announcement of $q$ to agent $b$. This increases the time at the
actual world by $2$, as in
Example~\ref{example:2-time-increase}. However, in the present
example, agent $a$ is uncertain whether the time has increased by $1$
or by $2$.

\begin{figure}
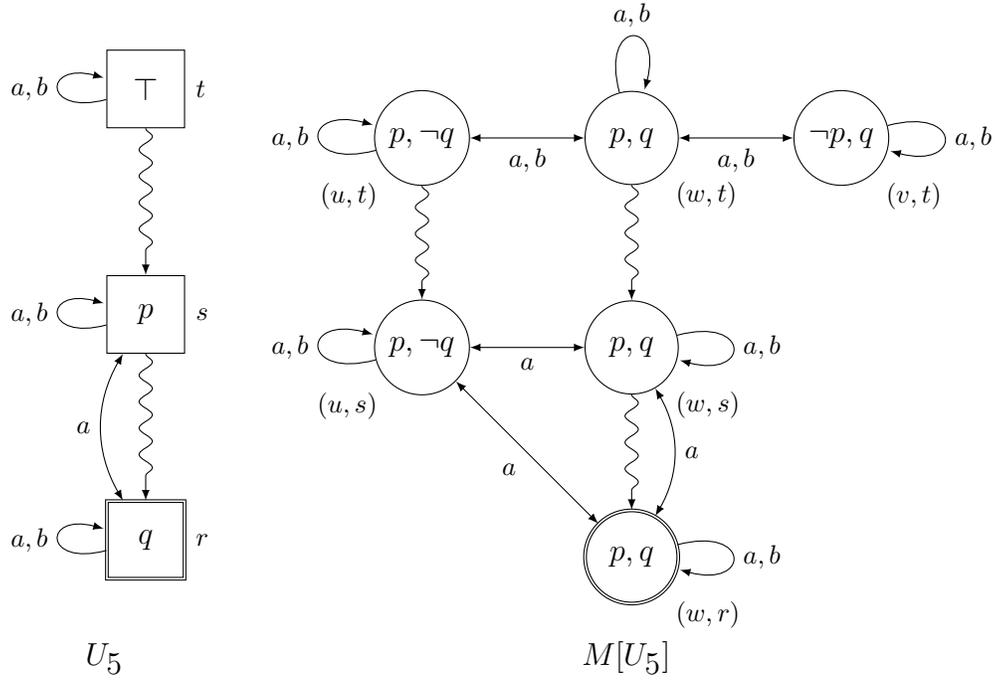

  \begin{center}
    \begin{tabular}{cc}
      \begin{mytikz}
        \node[e,label={right:$t$}] (t) {$\top$};
        
        \node[e,below of=t,label={right:$s$}] (s) {$p$};
        
        \node[e,double,below of=s,label={right:$r$}] (r) {$q$};
        
        \path (t) edge[loop left] node{$a,b$} ();
        
        \path (t) edge[->,zz] node{} (s);
        
        \path (s) edge[->,zz] node{} (r);
        
        \path (s) edge[<->,bend right=30, swap] node{$a$} (r);
        
        \path (s) edge[loop left] node{$a,b$} ();
        
        \path (r) edge[loop left] node{$a,b$} ();
      \end{mytikz}
      &
      \begin{mytikz}
        \node[w,xshift=2em,right of=t,label={below left:$(u,t)$}] (ut) {$p, \lnot q$};
        
        \node[w,right of=ut,label={below right:$(w,t)$}] (wt) {$p, q$};
        
        \node[w,right of=wt,label={below right:$(v,t)$}] (vt) {$\lnot p, q$};
        
        \path (wt) edge[loop above] node{$a,b$} ();
        
        \path (wt) edge[<->,swap] node{$a,b$} (vt);
        
        \path (wt) edge[<->] node{$a,b$} (ut);
        
        \path (vt) edge[loop right] node{$a,b$} ();
        
        \path (ut) edge[loop left] node{$a,b$} ();
        
        \node[w,below of=wt,label={below right:$(w,s)$}] (ws) {$p, q$};
        
        \node[w,below of=ut,label={below left:$(u,s)$}] (us) {$p, \lnot q$};
        
        \path (ws) edge[loop right] node{$a,b$} ();
        
        \path (us) edge[loop left] node{$a,b$} ();
        
        \path (ws) edge[<->] node{$a$} (us);
        
        \path (wt) edge[zz,->] node{} (ws);
        
        \path (ut) edge[zz,->] node{} (us);
        
        \node[w,double,below of=ws,label={below right:$(w,r)$}] (wr) {$p, q$};

        \path (wr) edge[loop right] node{$a, b$} ();
        
        \path (ws) edge[zz,->] node{} (wr);
        
        \path (wr) edge[<->,swap,bend right=30] node{$a$} (ws);
        
        \path (wr) edge[<->] node{$a$} (us);
        
      \end{mytikz}
      \\
      $U_{\ref{figure:2-time-increase-multi}}$ &
      $M[U_{\ref{figure:2-time-increase-multi}}]$
    \end{tabular}
  \end{center}
  \caption{Left: $(U_{\ref{figure:2-time-increase-multi}},r)$, the
    sequenced public announcement of $p$ followed by the asynchronous
    semi-private announcement of $q$ to agent $b$.  Right: the resulting
    situation $(M[U_{\ref{figure:2-time-increase-multi}}],
    (w,r))$. Agent arrows $\to_x$ are here implicitly closed under
    transitivity. (Example \ref{example:2-time-increase-multi})}
  \label{figure:2-time-increase-multi}
\end{figure}
\end{example}

We contrast Example~\ref{example:2-time-increase-multi} with the
following.

\begin{example}
  \label{example:1-time-increase-multi}
  $(U_{\ref{figure:1-time-increase-multi}},r)$
  (Figure~\ref{figure:1-time-increase-multi}) couples a public
  announcement of $p$ with the simultaneous semi-private announcement of
  $q$ to agent $b$.  When we compare this with Example
  \ref{example:2-time-increase-multi} (pictured in
  Figure~\ref{figure:2-time-increase-multi}), we note that the agents'
  respective knowledge gain is identical with respect to the
  propositional facts: $a$ learns that $p$ is true but not whether $q$
  is true, while $b$ learns both $p$ and $q$. However, in the current
  example (pictured in Figure~\ref{figure:1-time-increase-multi}), $b$
  learns $p$ and $q$ simultaneously instead of successively, and $a$'s
  knowledge differs accordingly.

  \begin{figure}
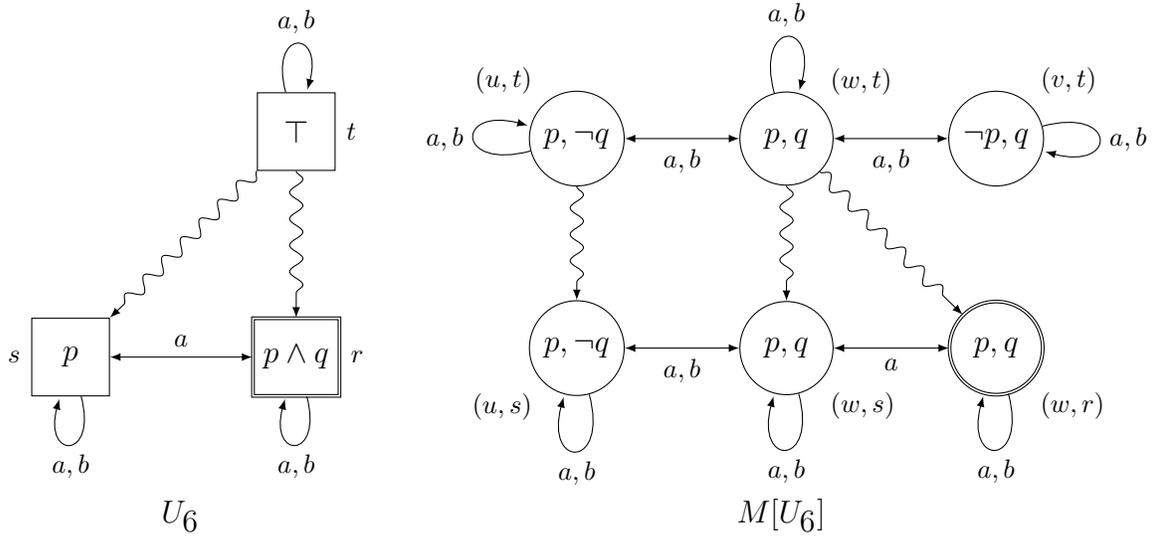

    \begin{center}
      \begin{tabular}{cc}
        \begin{mytikz}
          \node[e,label={right:$t$}] (t) {$\top$};
          
          \node[e,below of=t, double, label={right:$r$}] (r) {$p \wedge q$};
          
          \node[e,left of=r,label={left:$s$}] (s) {$p$};
          
          \path (t) edge[loop above] node{$a,b$} ();
          
          \path (t) edge[zz,->] node{} (s);
          
          \path (t) edge[zz,->] node{} (r);
          
          \path (r) edge[loop below] node{$a,b$} ();

          \path (s) edge[loop below] node{$a,b$} ();
          
          \path (s) edge[<->] node{$a$} (r);
        \end{mytikz}
        &
        \begin{mytikz}
          \node[w,xshift=2em,label={above left:$(u,t)$}] (ut) {$p, \lnot q$};
          
          \node[w,right of=ut,label={above right:$(w,t)$}] (wt) {$p, q$};
          
          \node[w,right of=wt,label={above right:$(v,t)$}] (vt) {$\lnot p, q$};
          
          \path (wt) edge[loop above] node{$a,b$} ();
          
          \path (wt) edge[<->,swap] node{$a,b$} (vt);
          
          \path (wt) edge[<->] node{$a,b$} (ut);
          
          \path (vt) edge[loop right] node{$a,b$} ();
          
          \path (ut) edge[loop left] node{$a,b$} ();
          
          \node[w,double,below of=vt,label={below right:$(w,r)$}] (wr) {$p, q$};
          
          \node[w,below of=wt,label={below right:$(w,s)$}] (ws) {$p, q$};
          
          \node[w,below of=ut,label={below left:$(u,s)$}] (us) {$p, \lnot q$};
          
          \path (ws) edge[loop below] node{$a,b$} ();
          
          \path (us) edge[loop below] node{$a,b$} ();
          
          \path (wr) edge[loop below] node{$a,b$} ();
          
          \path (ws) edge[<->] node{$a,b$} (us);
          
          \path (wr) edge[<->] node{$a$} (ws);
          
          \path (wt) edge[zz,->] node{} (ws);
          
          \path (wt) edge[zz,->] node{} (wr);
          
          \path (ut) edge[zz,->] node{} (us);
          
        \end{mytikz}
        \\
        $U_{\ref{figure:1-time-increase-multi}}$ &
        $M[U_{\ref{figure:1-time-increase-multi}}]$
      \end{tabular}
    \end{center}
    \caption{Left: $(U_{\ref{figure:1-time-increase-multi}},r)$, the
      public announcement of $p$ coupled with the simultaneous synchronous
      semi-private announcement of $q$ to agent $b$.  Right: the resulting
      situation $(M[U_{\ref{figure:1-time-increase-multi}}],
      (w,r))$. Agent arrows $\to_x$ are here implicitly closed under
      transitivity. (Example \ref{example:1-time-increase-multi})}
    \label{figure:1-time-increase-multi}
  \end{figure}
\end{example}

We now compare the truth values of several propositions in
$M[U_{\ref{figure:2-time-increase-multi}}]$ from
Example~\ref{example:2-time-increase-multi}
(Figure~\ref{figure:2-time-increase-multi}) and
$M[U_{\ref{figure:1-time-increase-multi}}]$ from
Example~\ref{example:1-time-increase-multi}
(Figure~\ref{figure:1-time-increase-multi}). First, it should be clear
from inspection that at both actual (i.e., double circled) worlds,
agent $a$ knows $p$ is true but does not know whether $q$ is true,
while agent $b$ knows both $p$ and $q$. Agent $a$ also considers $b$'s
current epistemic state \emph{vis-\`{a}-vis} $p$ and $q$ to be
possible, as both worlds satisfy $\diamondsuit_a \Box_b(p\land q)\land
\Box_b(p\land q)$; that is, agent $a$ considers it possible that $b$
knows $p$ and $q$, and this possibility is in fact the correct one,
since $b$ does in fact does know $p$ and $q$.  However, note that
$M[U_{\ref{figure:2-time-increase-multi}}], (w,r) \models
\diamondsuit_a [Y] \Box_b(p\land q)$ while
$M[U_{\ref{figure:1-time-increase-multi}}], (w,r) \not \models
\diamondsuit_a [Y] \Box_b(p\land q)$. The formula $\diamondsuit_a [Y]
\Box_b(p\land q)$ says that $a$ considers it possible that $b$ knew
$p$ and $q$ ``yesterday'' (i.e., one step in the past). So despite the
fact that agent $a$'s knowledge of $b$'s epistemic state
\emph{vis-\`{a}-vis} $p$ and $q$ is the same at the resultant
situations $(M[U_{\ref{figure:2-time-increase-multi}}],(w,r))$ and
$(M[U_{\ref{figure:1-time-increase-multi}}],(w,r))$, there is a key
difference: agent $a$'s knowledge of how these epistemic states
evolved over time is not the same.  In the first case
$(M[U_{\ref{figure:2-time-increase-multi}}],(w,r))$, agent $a$ thinks
$b$ may have learned $p$ before learning $q$.  However, in the second
case $(M[U_{\ref{figure:1-time-increase-multi}}],(w,r))$, though $a$
still thinks $b$ may have learned $q$, he believes that the only way
this could have happened is that $b$ learned both $q$ and $p$
simultaneously.


\section{Proof System and Completeness}
\label{section:proof-system}

\begin{figure}
  \begin{center}
    \textsc{Axiom Schemes}\\[.2em]
    \renewcommand{\arraystretch}{1.3}
    \begin{tabular}[t]{ll}
      CL. & Schemes for Classical Propositional Logic \\

      $K_a$. & $\Box_a(\varphi\imp\psi)\imp(\Box_a\varphi\imp\Box_a\psi)$ \\

      $K_Y$. &
      $[Y](\varphi\imp\psi)\imp([Y]\varphi\imp[Y]\psi)$ \\

      UA. & $[U,s]q \iff \bigl( \pre^U(s)\imp q \bigr)$
      for $q\in\Prop$ \\

      U$\land$. & $[U,s](\varphi\land\psi) \iff
      \bigl( [U,s]\varphi \land [U,s]\psi \bigr)$
      \\

      U$\lnot$. & $[U,s]\lnot\varphi \iff
      \bigl( \pre^U(s)\imp\lnot[U,s]\varphi \bigr)$ \\

      U$\Box_a$. & $[U,s]\Box_a\varphi \iff
      \bigl(
      \pre^U(s) \imp
      \bigwedge_{s'\from^U_a s }\Box_a[U,s']\varphi
      \bigr)$ \\

      U$[Y]$. & 
      $[U,s][Y]\varphi \iff  \bigl( \pre^U(s) \imp [Y][U,s]\varphi
      \bigr)$
      if $s$ is a past state
      \\
      &
      $[U,s][Y]\varphi \iff
      \bigl( \pre^U(s)
      \imp\bigwedge_{s'\leadsto^U s}[U,s']\varphi  \bigr)$
      if $s$ is not a past state
    \end{tabular}
    \\[1.3em]
    \textsc{Rules}
    \vspace{-1em}
    \[
    \begin{array}{c}
      \varphi\imp\psi \quad \varphi
      \\\hline
      \psi
    \end{array}
    \,\text{\footnotesize (MP)}
    \quad
    \begin{array}{c}
      \varphi
      \\\hline
       \Box_a\varphi
    \end{array}
    \,\text{\footnotesize (MN)}
    \quad
    \begin{array}{c}
      \varphi
      \\\hline
      [Y]\varphi
    \end{array}
    \,\text{\footnotesize (YN)}
    \quad
    \begin{array}{c}
      \varphi
      \\\hline
       [U,s]\varphi
    \end{array}
    \,\text{\footnotesize (UN)}
    \]
  \end{center}
  \caption{The theory $\DETL$}
  \label{figure:theory}
\end{figure}

\begin{definition}
  The \emph{axiomatic theory of Dynamic Epistemic Temporal Logic},
  $\DETL$, is defined in Figure~\ref{figure:theory}.
\end{definition}

Many axioms of $\DETL$ are the same as in Dynamic Epistemic
Logic.\footnote{\label{footnote:explaining-proof-system}The primary
  difference is with those concerning the $Y$-modality. There are two
  cases: $s$ is a past state, and $s$ is not a past state.  In both
  cases, U$[Y]$ is a simplification of U$\Box_a$, reflecting the
  involvement of asynchronous composition rather than synchronous (see
  the discussion after Definition~\ref{definition:semantics}). In
  U$\Box_a$, the conjunction reflects the transitions made in the
  action model, while the modality $\Box_a$ that follows reflects the
  transitions made in original model. Note that if $s$ is a past
  state, then there is no $s'\leadsto^U s$, so we can remove the
  conjunction. If $s$ not a past state, then it is the first
  coordinate rather than the second coordinate that must be fixed in a
  $\leadsto$ transition in the updated model. Hence we remove the
  modality $[Y]$ that would otherwise follow the conjunction.}  In
defining $\DETL$, we have not imposed any of the properties from
Definition~\ref{definition:model-properties} on action models, nor
have we designed the axiomatics to be sound for Kripke models having
properties from Definition~\ref{definition:model-properties} that one
might expect.  So $\DETL$ should be viewed as the \emph{minimal}
theory that brings update mechanisms to a basic Epistemic Temporal
Logic. However, we will study the preservation of these properties in
Section~\ref{section:preservation}, and we study a $\DETL$-based
theory satisfying a number of these properties in
Section~\ref{section:connections}.

\begin{theorem}[$\ldetl$ Reduction]
  \label{theorem:ldetl-reduction}
  For every $\ldetl$-formula $\varphi$, there is an action model-free
  $\lsetl$-formula $\varphi^\circ$ such that
  $\DETL\vdash\varphi\iff\varphi^\circ$.
\end{theorem}
\begin{proof}
  The proof is a straightforward adaptation of the standard argument
  from Dynamic Epistemic Logic \cite{DitHoeKoo07}.
\end{proof}

\begin{theorem}[Soundness and Completeness]
  \label{theorem:soundness-completeness}
  $\DETL\vdash\varphi$ if and only if $\models\varphi$.
\end{theorem}

\begin{proof}
  Soundness ($\vdash\varphi$ implies $\models\varphi$) is by induction
  on the length of $\DETL$-derivations.  All cases except U$[Y]$
  soundness are straightforward adaptations of the standard atemporal
  Dynamic Epistemic Logic arguments \cite{DitHoeKoo07}, so we shall
  only prove $U[Y]$ soundness here.  Proceeding, we are to show that
  \begin{eqnarray*}
    \models [U,s][Y]\varphi &\iff&  
    \bigl( \pre^U(s) \imp [Y][U,s]\varphi \bigr)
    \text{ if $s$ is a past state, and}
    \label{eq:s-is-ps}
    \\{}
    \models [U,s][Y]\varphi &\iff&
    \textstyle
    \bigl( \pre^U(s)
    \imp\bigwedge_{s'\leadsto^U s}[U,s']\varphi  \bigr)
    \text{ if $s$ is not a past state.}
    \label{eq:s-not-ps}
    \label{eq:correctness-pf}
  \end{eqnarray*}
  Given $(M,w)$, we assume that $M,w\models\pre^U(s)$, for otherwise
  the result follows immediately.
  \begin{itemize}
  \item Case:  $s$ is a past state.

    Assume $M,w\models[U,s][Y]\varphi$.  By the definition of truth,
    $M[U],(w',s)\models\varphi$ for each $(w',s)\leadsto^{M[U]}(w,s)$.
    Therefore, if $v\leadsto^M w$ satisfies $M,v\models\pre^U(s)$,
    then $M[U],(v,s)\models\varphi$. Conclusion:
    $M,w\models[Y][U,s]\varphi$.

    Conversely, assume $M,w\models[Y][U,s]\varphi$ and
    $(w',s)\leadsto^{M[U]}(w,s)$.  The second assumption implies both
    that $w'\leadsto^M w$---and hence $M,w'\models[U,s]\varphi$ by the
    first assumption---and that $M,w'\models\pre^U(s)$.  But then
    $M[U],(w',s)\models\varphi$.  Conclusion:
    $M,w\models[U,s][Y]\varphi$.

  \item Case: $s$ is not a past state.

    Assume $M,w\models[U,s][Y]\varphi$. By the definition of truth,
    $M[U],(w,t)\models\varphi$ for each $(w,t)\leadsto^{M[U]}(w,s)$.
    If $s'\leadsto^Us$ satisfies $M,w\models\pre^U(s')$, then
    $(w,s')\leadsto^{M[U]}(w,s)$ and hence
    $M[U],(w,s')\models\varphi$. Conclusion:
    $M,w\models\bigwedge_{s'\leadsto^U s}[U,s']\varphi$.

    Conversely, assume $M,w\models\bigwedge_{s'\leadsto^U
      s}[U,s']\varphi$ and $(w,t)\leadsto^{M[U]}(w,s)$.  The second
    assumption implies both that $t\leadsto^U s$---and hence
    $M,w\models[U,t]\varphi$ by the first assumption---and that
    $M,w\models\pre^U(t)$.  But then
    $M[U],(w,t)\models\varphi$. Conclusion:
    $M,w\models[U,s][Y]\varphi$.
  \end{itemize}

  Completeness ($\nvdash\varphi$ implies $\not\models\varphi$) follows
  by $\ldetl$ Reduction (Theorem~\ref{theorem:ldetl-reduction}), the
  standard normal modal logic canonical model argument for the action
  model-free sublanguage $\lsetl$ \cite{BlaRijVen01}, and the
  combination of $\ldetl$ Reduction with soundness.
\end{proof}

\section{Preservation Results}
\label{section:preservation}

In this section, we study the preservation of properties of Kripke
models defined previously in
Definition~\ref{definition:model-properties}.  These properties have
been of interest in the study of time in Dynamic Epistemic Logic
\cite{BenGerHos09,BenGerPac07,deglowwit11,Sac10,Sac08,Yap07} and so it
will be useful for our purposes to understand the conditions under
which they are preserved within our $\DETL$ setting.  Theorem
\ref{theorem:past-state} concerns the behavior of past states in
action models, and Theorem \ref{theorem:preservation} concerns the
preservation of Kripke model properties.

\begin{theorem}[Past State]
  \label{theorem:past-state}
  Let $(U,s)$ be an action satisfying $M,w\models\pre^U(s)$.
  \begin{enumerate}
    \renewcommand{\theenumi}{(\alph{enumi})}
    \renewcommand{\labelenumi}{\theenumi}
  \item \label{item:theorem:past-state:bounded-morphism}\label{item:ps:a}
    If $s$ is an epistemic past state, then $(M[U],(w,s))$ and $(M,w)$
    satisfy the same $\ldetl$-formulas.

  \item
  \label{item:theorem:past-state:truth}\label{item:ps:b}
  If $(U,s)$ is past preserving, then there is a history
  \[
  (w,s_0) \leadsto^{M[U]} (w,s_1) \leadsto^{M[U]} \cdots
  \leadsto^{M[U]} (w,s_n) = (w,s)
  \]
  such that $(M[U],(w,s_0))$ and $(M,w)$ satisfy the same
  $\ldetl$-formulas.
  \end{enumerate}
\end{theorem}
\begin{proof}
  \ref{item:ps:a} Since the language of $\ldetl$ does not include
  forward-looking tomorrow operators $[T]$, it follows by the standard
  argument in modal logic \cite{BlaRijVen01} that bisimilar worlds
  satisfy the same action model-free $\ldetl$-formulas (i.e., the same
  $\lsetl$-formulas).\footnote{In detail: a \emph{bisimulation} is a
    nonempty binary relation $B$ between the worlds of Kripke models
    $M=(W^M,\to^M,\leadsfrom^M,V^M)$ and
    $M'=(W^{M'},\to^{M'},\leadsfrom^{M'},V^{M'})$ with yesterday such
    that $wBw'$ implies $w$ and $w'$ satisfy the same propositional
    letters; and, for each binary relation symbol $R\in\{{\to_a}\mid
    a\in\Agnt\}\cup\{\leadsfrom\}$ specified by the structures, $wBw'$
    and $wR^Mv$ implies there is a $v'\in W^{M'}$ such that
    $w'R^{M'}v'$ and $vBv'$, and $wBw'$ and $w'R^{M'}v'$ implies there
    is a $v\in W^M$ such that $wR^Mv$ and $vBv'$.  Note that the
    definition of bisimulation only considers temporally reachable
    worlds ``in the past'' (i.e., in the direction from $x$ to $y$ in
    the arrow $x\leadsfrom y$).}  By $\ldetl$ Reduction
  (Theorem~\ref{theorem:ldetl-reduction}) and soundness
  (Theorem~\ref{theorem:soundness-completeness}), bisimilar worlds
  also satisfy the same $\ldetl$-formulas. Finally, one can show that
  there is a bisimulation between $(w,s)$ and $w$.  The result
  follows.

  \ref{item:ps:b} Past preservation of $(U,s)$ implies there exists a
  history $s_0\leadsto^U\cdots\leadsto^Us_n=s$ that begins at an
  epistemic past state $s_0$. The result therefore follows by part
  \ref{item:ps:a}.
\end{proof}

Theorem~\ref{theorem:past-state}\ref{item:ps:a} tells us that, from
the point of view of the language, an epistemic past state essentially
makes a copy of a given situation $(M,w)$ within the context of the
updated model $M[U]$.  Part \ref{theorem:past-state}\ref{item:ps:b}
tells us that after the execution of a past preserving action $(U,s)$,
the copy of the initial situation resides at the beginning of a
history
\begin{equation}
  (w,s_0)\leadsto^{M[U]}\cdots\leadsto^{M[U]}(w,s_n)=(w,s)
  \label{eq:produced-history}
\end{equation}
produced by the stepwise occurrence of a past state $s_0$ followed by
events $s_1,\ldots,s_n=s$.  Past states in past preserving actions
$(U,s)$ therefore play the role of ``maintaining a link to the past''
in virtue of the fact that there is a temporal linkage
\eqref{eq:produced-history} in the resultant model $M[U]$ leading back
to the initial situation (as it exists in its copied form).

Theorem~\ref{theorem:past-state} would fail if we were to include a
tomorrow operator $[T]$ in our language.\footnote{\label{footnote:consequence-of-T-to-past-state-thm}For 
  example, we have
  $M,w\not\models\may{T}p$ (Figure~\ref{figure:M}) and yet
  $M[U_{\ref{figure:no-time-increase}}],(w,t)\models\may{T}p$
  (Figure~\ref{figure:no-time-increase}).} This not only violates the
theorem but also our intuition about what it means for the occurrence
of an action to increment the time.  In particular, given that we
identify the time of a world with its depth, a time-incrementing
action must take an initial situation and successively add on
additional layers of ``future'' worlds:
\begin{equation}
  \text{(initial world)} \leadsto
  \text{($+1$ world)} \leadsto
  \text{($+2$ world)} \leadsto \cdots
  \leadsto
  \text{($+n$ world)}\enspace.
  \label{eq:progression}
\end{equation}
In order for us to view this process as a progression that began with
a particular situation $(M,w)$, the leftmost world in the temporal
sequence \eqref{eq:progression} should be identical to our initial
situation $(M,w)$ from the point of view of the language.
Furthermore, we should be able to ``trace backward in time'' from the
final resultant situation---made up of the rightmost world in the
temporal sequence \eqref{eq:progression}---to our initial situation.
And this is just what Theorem~\ref{theorem:past-state} says we can do.
So since adding the $[T]$ operator would falsify the theorem and hence
go against our intuition, we decided to leave this operator out.

We now examine the relationship between our conditions on action
models (Definition~\ref{definition:model-properties}) and the
preservation of certain Kripke models properties (also
Definition~\ref{definition:model-properties}) under the update
operation $M\mapsto M[U]$.

\begin{theorem}[Preservation]
  \label{theorem:preservation}
  Suppose $M,w^*\models\pre^U(s^*)$ for some $w^*\in W^M$ and $s^*\in
  W^U$.
  \begin{enumerate}
    \renewcommand{\theenumi}{(\alph{enumi})}
    \renewcommand{\labelenumi}{\theenumi}
  \item \label{item:pres-persist} If $M$ satisfies persistence of facts, then so does $M[U]$.

  \item \label{item:pres-DD} If $M$ and $U$ are depth-defined, then so is $M[U]$.

  \item \label{item:pres-pastK} If $M$ and $U$ satisfy knowledge of the past and $U$ is
    history preserving, then $M[U]$ satisfies knowledge of the past.

  \item \label{item:pres-initK} If $M$ satisfies knowledge of the initial time and $U$ is
    history preserving, then $M[U]$ satisfies knowledge of the initial
    time.

  \item \label{item:pres-past-uniq} If $M$ and $U$ satisfy uniqueness of the past, then so does
    $M[U]$.

  \item \label{item:pres-PR} If $M$ and $U$ satisfy perfect recall and $U$ is history
    preserving, then $M[U]$ satisfies perfect recall.

  \item \label{item:pres-sync} If $M$ and $U$ are synchronous and $U$ is history preserving,
    then $M[U]$ is synchronous.
  \end{enumerate}
\end{theorem}
\begin{proof}
  We prove each item in turn. \vspace{1.5ex}
  
 \emph{\ref{item:pres-persist} If $M$ satisfies persistence of facts, then so does $M[U]$.}
    Suppose that $M$ satisfies persistence of facts and $(w,s)
    \leadsto^{M[U]} (w',s')$.  It follows that $w=w'$ or $w\leadsto^M
    w'$.  Now $M$ satisfies persistence of facts, so $w\in V^M(p)$ iff
    $w'\in V^M(p)$.  Applying the fact that $(v,t)\in V^{M[U]}(p)$ iff
    $v\in V^M(p)$, we have $(w,s)\in V^{M[U]}(p)$ iff $(w',s')\in
    V^{M[U]}(p)$.\vspace{1.7ex}

  \emph{\ref{item:pres-DD} If $M$ and $U$ are depth-defined, then so is $M[U]$.}
    Assume that $M$ and $U$ are depth-defined. Notice that if
    \begin{equation}
      \forall(w,s)\in W^{M[U]}: \mydepth(w,s)\leq
      \mydepth(w)+\mydepth(s)\enspace,
      \label{eq:pres:dd}
    \end{equation}
    then since we have $\mydepth(w)<\infty$ and $\mydepth(s)<\infty$
    by the depth-definedness of $M$ and $U$, it follows that $M[U]$ is
    depth-defined.  It therefore suffices to prove \eqref{eq:pres:dd}
    by induction on $\mydepth(s)$.
    \begin{itemize}
    \item Base case: $\mydepth(s)=0$.
      
      It follows that $s$ is a past state.  Therefore
      $(w',s')\leadsto^{M[U]}(w,s)$ implies $w'\leadsto^M w$ and
      $s'=s$. We now prove \eqref{eq:pres:dd} by a sub-induction on
      $\mydepth(w)$.  In the sub-induction base case, $\mydepth(w)=0$
      and therefore there is no $w'\leadsto^M w$, which implies there
      is no $(w',s)\leadsto^{M[U]}(w,s)$.  But then
      $\mydepth(w,s)=0=\mydepth(w)+\mydepth(s)$, which completes the
      sub-induction base case.  For the sub-induction step, we assume
      that \eqref{eq:pres:dd} holds for all worlds $v$ having
      $0\leq\mydepth(v)<\mydepth(w)$ (the ``sub-induction
      hypothesis'') and we prove \eqref{eq:pres:dd} holds for world
      $w$ itself.  If $\mydepth(w,s)=0$, then \eqref{eq:pres:dd}
      follows immediately because depths are non-negative integers.
      So let us assume that $\mydepth(w,s)>0$.  Then we may choose an
      arbitrary $(w',s)\leadsto^{M[U]}(w,s)$, from which it follows
      that $w'\leadsto^M w$. Since $M$ is depth-defined,
      $\mydepth(w')\leq\mydepth(w)-1$ and so we may apply the
      sub-induction hypothesis:
      \[
      \mydepth(w',s) \leq \mydepth(w')+\mydepth(s) \leq
      \mydepth(w)-1+\mydepth(s)\enspace.
      \]
      Hence
      \begin{eqnarray*}
        \mydepth(w,s) &=&
        1+\max\{\mydepth(w',s)\mid 
        (w',s)\leadsto^{M[U]}(w,s)\}
        \\
        &\leq&
        1+\max\{\mydepth(w)-1+\mydepth(s) \mid
        (w',s)\leadsto^{M[U]}(w,s)\}
        \\
        &=&
        \mydepth(w)+\mydepth(s)
      \end{eqnarray*}

    \item Induction step: we suppose \eqref{eq:pres:dd} holds for all
      events $t$ having $0\leq\mydepth(t)<\mydepth(s)$ (the
      ``induction hypothesis'') and prove that \eqref{eq:pres:dd}
      holds for event $s$ itself.

      $s$ is not a past state because $\mydepth(s)>0$.  Therefore
      $(w',s')\leadsto^{M[U]}(w,s)$ implies $w'=w$ and $s'\leadsto^U
      s$.  If $\mydepth(w,s)=0$, then \eqref{eq:pres:dd} follows
      immediately because depths are non-negative integers.  So let us
      assume $\mydepth(w,s)>0$. Then we may choose an arbitrary
      $(w,s')\leadsto^{M[U]}(w,s)$, from which it follows that
      $s'\leadsto^U s$. Since $U$ is depth-defined, we have
      $\mydepth(s')\leq\mydepth(s)-1$ and so we may apply the
      induction hypothesis:
      \[
      \mydepth(w,s') \leq \mydepth(w)+\mydepth(s') \leq
      \mydepth(w)+\mydepth(s)-1\enspace.
      \]
      Hence
      \begin{eqnarray*}
        \mydepth(w,s) &=&
        1+\max\{\mydepth(w,s')\mid 
        (w,s')\leadsto^{M[U]}(w,s)\}
        \\
        &\leq&
        1+\max\{\mydepth(w)+\mydepth(s)-1 \mid
        (w,s')\leadsto^{M[U]}(w,s)\}
        \\
        &=&
        \mydepth(w)+\mydepth(s)
      \end{eqnarray*}
    \end{itemize}
\vspace{1ex}

  \emph{\ref{item:pres-pastK} If $M$ and $U$ satisfy knowledge of the past and $U$ is
    history preserving, them $M[U]$ satisfies knowledge of the past.}
    Suppose $M$ and $U$ satisfy knowledge of the past, $U$ is history
    preserving, and
    \[
    (w',s')\leadsto^{M[U]} (w,s)\to^{M[U]}_a (v,t)\enspace.
    \]
    We want to show that there exists $(v',t') \leadsto^{M[U]}
    (v,t)$. Given $(w,s) \to^{M[U]}_a (v,t)$, we have $w \to^M_a v$
    and $s \to^U_a t$.  Given $(w',s')\leadsto^{M[U]} (w,s)$, one of
    two cases obtains.
    \begin{itemize}
    \item Case: $w' \leadsto^M w$ and $s'=s$ is a past state.

      Since $w'\leadsto^M w \to^M_a v$ and $M$ satisfies knowledge of
      the past, there exists $v' \leadsto^M v$. Since $U$ is history
      preserving and $s$ is a past state, $s$ is an epistemic past
      state.  From this we obtain two things.  First, $(v',s)\in
      W^{M[U]}$ because epistemic past states have valid
      preconditions.  Second, applying the fact that $s \to^U_a t$, it
      follows that $s=t$ because $\to_a$ arrows leaving epistemic past
      states are all reflexive.  Since $v' \leadsto^M v$ and $s$ is a
      past state, we conclude that $(v',s)\leadsto^{M[U]}
      (v,s)=(v,t)$.

    \item Case: $w'=w$ and $s'\leadsto^U s$.

      Since $s'\leadsto^U s \to^U_a t$ and $U$ satisfies knowledge of
      the past, there exists $t' \leadsto^U t$. Applying this to the
      assumption that $U$ is history preserving and the fact that
      $(v,t)\in W^{M[U]}$, it follows that $(v,t') \in W^{M[U]}$. But
      then $(v,t') \leadsto^{M[U]} (v,t)$.
    \end{itemize}
\vspace{1ex}

  \emph{\ref{item:pres-initK} If $M$ satisfies knowledge of the initial time and $U$ is
    history preserving, then $M[U]$ satisfies knowledge of the initial
    time.}
    Suppose $M$ satisfies knowledge of the initial time, $U$ is
    history preserving, $(w,s) \to^{M[U]}_a (v,t)$, and
    $\mydepth(w,s)=0$. We wish to show that $\mydepth(v,t)=0$ as well.
    Toward a contradiction, assume $\mydepth(v,t)\neq0$, which implies
    there exists $(v',t') \leadsto^{M[U]} (v,t)$.  It follows from
    $(w,s) \to^{M[U]}_a (v,t)$ that $w \to^M_a v$ and $s \to^U_a t$.
    We consider two cases.
    \begin{itemize}
    \item Case: $s$ is not a past state.

      Since $s$ is not a past state, there exists $s' \leadsto^U
      s$. Since $U$ is history preserving and $(w,s)\in W^{M[U]}$, it
      follows that $(w,s')\in W^{M[U]}$. But then $(w,s')
      \leadsto^{M[U]} (w,s)$, which contradicts our assumption that
      $\mydepth(w,s)=0$.

    \item Case: $s$ is a past state.

      Since $U$ is history preserving and $s$ is a past state, $s$ is
      in fact an epistemic past state. Applying the fact that $s
      \to^U_a t$, it follows that $s=t$ because $\to_a$ arrows leaving
      epistemic past states are all reflexive.  Since
      $(v',t')\leadsto^{M[U]} (v,t)$ and $t=s$ is a past state, it
      follows that $t'=t=s$ and $v\leadsto^M v'$.  Also, it follows
      from the fact that $s$ is a past state and $\mydepth(w,s)=0$
      that we have $\mydepth(w)=0$.  Since $M$ satisfies knowledge of
      the initial time, it follows from $\mydepth(w)=0$ and $w\to^M_a
      v$ that $\mydepth(v)=0$, but this contradicts $v\leadsto^M v'$.
    \end{itemize}
    Since both cases lead to a contradiction, we conclude that
    $\mydepth(v,t)=0$, as desired. \vspace{1.7ex}

  \emph{\ref{item:pres-past-uniq} If $M$ and $U$ satisfy uniqueness of the past, then so does
    $M[U]$.}
    Suppose $M$ and $U$ satisfy uniqueness of the past and we have
    $(v_1,t_1)\leadsto^{M[U]}(w,s)$ and
    $(v_2,t_2)\leadsto^{M[U]}(w,s)$.  We wish to show that
    $(v_1,t_1)=(v_2,t_2)$.  There are two cases to consider.
    \begin{itemize}
    \item Case:  $s$ is a past state.

      Since $s$ is a past state, if follows from
      $(v_1,t_1)\leadsto^{M[U]}(w,s)$ and
      $(v_2,t_2)\leadsto^{M[U]}(w,s)$ that we have $s=t_1=t_2$,
      $v_1\leadsto^M w$, and $v_2\leadsto^M w$.  Since $M$ satisfies
      uniqueness of the past, it follows that $v_1=v_2$.  Hence
      $(v_1,t_1)=(v_2,t_2)$.

    \item Case: $s$ is not a past state.

      Since $s$ is not a past state, if follows from
      $(v_1,t_1)\leadsto^{M[U]}(w,s)$ and
      $(v_2,t_2)\leadsto^{M[U]}(w,s)$ that we have $w=v_1=v_2$,
      $t_1\leadsto^U s$, and $t_2\leadsto^U s$.  Since $U$ satisfies
      uniqueness of the past, it follows that $t_1=t_2$. Hence
      $(v_1,t_1)=(v_2,t_2)$.
    \end{itemize}
    \vspace{1ex}

  \emph{ \ref{item:pres-PR} If $M$ and $U$ satisfy perfect recall and $U$ is history
    preserving, then $M[U]$ satisfies perfect recall.}
    Suppose $M$ and $U$ satisfy perfect recall, $U$ is history
    preserving, and $(w',s') \leadsto^{M[U]} (w,s) \to^{M[U]}_a
    (v,t)$.  We wish to prove that there exists $(v',t')\in W^{M[U]}$
    such that $(w',s')\to^{M[U]}_a (v',t')\leadsto^{M[U]}(v,t)$.
    Proceeding, it follows from $(w,s) \to^{M[U]}_a (v,t)$ that
    $w\to^M_a v$ and $s\to^U_a t$.  It follows from
    $(w',s')\leadsto^{M[U]}(w,s)$ that one of two cases obtains.
    \begin{itemize}
    \item Case: $w'=w$ and $s'\leadsto^U s$.

      Since $s' \leadsto^U s \to^U_a t$ and $U$ satisfies perfect
      recall, there exists $t'$ satisfying $s' \to^U_a t' \leadsto^U
      t$. Since $U$ is history preserving and $(v,t)\in W^{M[U]}$, it
      follows that $(v,t')\in W^{M[U]}$.  But then
      $(v,t')\leadsto^{M[U]}(v,t)$.  Since $w\to^M_av$ and
      $s'\to^U_at'$, we have $(w',s')=(w,s')\to^{M[U]}_a(v,t')$.

    \item Case: $w'\leadsto^M w$ and $s'=s$ is a past state.

      Since $w' \leadsto^M w \to^M_a v$ and $M$ satisfies perfect
      recall, there exists $v'$ satisfying $w' \to^M_a v' \leadsto^M
      v$. Further, since $s$ is a past state and $U$ is history
      preserving, $s$ is an epistemic past state.  From this two
      things follow.  First, $(v',s)\in W^{M[U]}$ because epistemic
      past states have valid preconditions.  Second, applying the fact
      that $s\to^U_a t$, it follows that $s=t$ because $\to_a$ arrows
      leaving epistemic past states are all reflexive.  Since
      $v'\leadsto^M v$ and $s$ is a past state, we have
      $(v',s)\leadsto^{M[U]}(v,s)=(v,t)$.  Further, since
      $w'\to^M_av'$ and $s\to^U_at=s$, we have
      $(w',s')=(w',s)\to^{M[U]}_a(v',s)$.
    \end{itemize}
    \vspace{1ex}

  \emph{\ref{item:pres-sync} If $M$ and $U$ are synchronous and $U$ is history preserving,
    then $M[U]$ is synchronous.}
    Assume $M$ and $U$ are synchronous and $U$ is history preserving.
    It follows by a previous item that $M[U]$ is depth-defined.  So
    all that remains is to prove that $(w,s)\to^{M[U]}_a(w',s')$
    implies $\mydepth(w,s)=\mydepth(w',s')$.  To prove this,
    we for the moment assume the following:
    \begin{equation}
      \forall(v,t)\in W^{M[U]}:
      \mydepth(v,t)=\mydepth(v)+\mydepth(t)\enspace.
      \label{eq:pres:sync}
    \end{equation}
    Proceeding under this assumption, it follows from
    $(w,s)\to^{M[U]}_a(w',s')$ that $w \to^M_a w'$ and $s \to^U_a
    s'$. Since $M$ and $U$ are synchronous, it follows that
    $\mydepth(w)=\mydepth(w')$ and $\mydepth(s)=\mydepth(s')$.
    But then
    \[
    \mydepth(w,s)=\mydepth(w)+\mydepth(s)=
    \mydepth(w')+\mydepth(s')=\mydepth(w',s')
    \]
    by \eqref{eq:pres:sync}, completing the argument.  So all that
    remains is to prove \eqref{eq:pres:sync}.  The proof proceeds by
    induction on $\mydepth(t)$.
    \begin{itemize}
    \item Base case: $\mydepth(t)=0$.  

      It follows that $t$ is a past state.  Since $U$ is history
      preserving, $t$ is an epistemic past state.  Therefore, $u\in
      W^M$ implies $(u,t)\in W^{M[U]}$ because epistemic past states
      have valid preconditions, and hence $u\leadsto^M u'$ implies
      $(u,t)\leadsto^{M[U]}(u',t)$.  It follows by an easy
      sub-induction on $\mydepth(v)$ that $\mydepth(v,t)=\mydepth(v)$.
      Since $\mydepth(t)=0$, this proves \eqref{eq:pres:sync}.

    \item Induction step: assume the result holds for events $s$ with
      $\mydepth(s)<\mydepth(t)$ (the ``induction hypothesis'') and
      prove the result holds for event $t$ with $\mydepth(t)>0$.

      Since $\mydepth(t)>0$, there exists $t'\leadsto^U t$ with
      $\mydepth(t')=\mydepth(t)-1$.  Since $U$ is history preserving,
      we have $\models\pre^U(t)\imp\pre^U(t')$. But then it follows
      from $(v,t)\in W^{M[U]}$ that $(v,t')\in W^{M[U]}$.  Since
      $\mydepth(t')<\mydepth(t)$, we may apply the induction
      hypothesis, from which it follows that
      \[
      \mydepth(v,t')=\mydepth(v)+\mydepth(t')=
      \mydepth(v)+\mydepth(t)-1\enspace.
      \]
      Further, we note that from $t'\leadsto^U t$ we obtain
      $(v,t')\leadsto^{M[U]}(v,t)$.  Now take an arbitrary
      $(u,s)\leadsto^{M[U]}(v,t)$.  Since $\mydepth(t)>0$, event $t$
      is not a past state and it therefore follows that $u=v$ and
      $s\leadsto^U t$.  Hence $\mydepth(s)\leq \mydepth(t)-1$.
      Applying the induction hypothesis,
      \[
      \mydepth(u,s)=\mydepth(v,s)=
      \mydepth(v)+\mydepth(s)\leq\mydepth(v)+\mydepth(t)-1\enspace.
      \]
      Note that this holds for all $(u,s)\leadsto^{M[U]}(v,t)$.
      Finally, since
      \[
      \mydepth(v,t)=1+\max\{\mydepth(u,s)\mid
      (u,s)\leadsto^{M[U]}(v,t)\}\enspace,
      \]
      it follows by what we have shown above that
      \[
      \mydepth(v,t)=1+\mydepth(v,t')=\mydepth(v)+\mydepth(t)\enspace,
      \]
      which completes the proof. \qedhere
    \end{itemize}
\end{proof}

Theorem~\ref{theorem:preservation} describes conditions under which
properties of epistemic temporal models are preserved under updates.
We will use this theorem later to show that a well-studied atemporal
Dynamic Epistemic Logic approach to reasoning about time is limited to
the class of Kripke models that \emph{necessarily satisfy all of the
  properties we have defined}.  This highlights one of the key
advantages of our $\DETL$ framework: it may be used to reason about
situations that do or do not satisfy these (or other) properties.  The
choice is left to the user.

\section{Connections with Previous Work}
\label{section:connections}

\subsection{$\RDETL$}
\label{section:rdetl}

Theorem~\ref{theorem:preservation} studied the preservation of Kripke
model properties under certain actions.  We chose these properties
because they have been of interest in many studies of time in Dynamic
Epistemic Logic
\cite{BenGerHos09,BenGerPac07,deglowwit11,Sac10,Sac08,Yap07}.  We now
focus our attention on the class of Kripke models that satisfy these
properties.  This provides a paradigmatic example demonstrating how
our $\DETL$ framework can be used to reason about a well-studied account
of time in Dynamic Epistemic Logic.

\begin{definition}[Restricted (forest-like) models]
\label{def:restricted}
A Kripke model $M$ is \emph{restricted} (or \emph{forest-like}) if it
satisfies persistence of facts, depth-definedness, knowledge of the
past, knowledge of the initial time, uniqueness of the past, and
perfect recall (Definition~\ref{definition:model-properties}).  Let
$\restricted$ be the class of all the restricted Kripke models and
$\restricted_\pt$ the class of all pointed restricted Kripke models.
\end{definition}
The restricted models satisfy all the constraints on Kripke models
given in Definition~\ref{definition:model-properties}. Although
synchronicity was not explicitly named as one of the properties of a
restricted model, it is not hard to show that synchronicity does
follow from the other properties (argue by induction on the depth of
worlds, making use of perfect recall, knowledge of the past, and
knowledge of the initial time).

We now define a fragment of $\ldetl$ whose update modals preserve
these restricted models.

\begin{definition}[Language $\lrdetl$]
  \label{definition:lrdetl}
  The language $\lrdetl$ of restricted $\DETL$ is the sublanguage of
  $\ldetl$ obtained by removing all actions $[U,s]$ that are based on
  an action model $U$ that fails to satisfy one or more of persistence
  of facts, depth-definedness, knowledge of the past, history
  preservation, knowledge of the initial time, uniqueness of the past,
  or perfect recall.  This removal applies recursively to
  preconditions as well.
\end{definition}

The restrictions on the action models in $\lrdetl$ are those that
appear in the Preservation Theorem
(Theorem~\ref{theorem:preservation}). Hence updating a restricted
model by an action model in $\Actm(\lrdetl)$ yields another restricted
model. 

\begin{definition}[$\RDETL$ Semantics]
  We write $M,w\models_{\RDETL}\varphi$ to mean that
  $(M,w)\in\restricted_*$ and $M,w\models\varphi$.  We write
  $\models_{\RDETL}\varphi$ to mean that $M,w\models\varphi$ for every
  $(M,w)\in \restricted_*$.
\end{definition}

\subsubsection{Proof system for $\RDETL$}
\label{section:soundness-completeness-YDEL-DETLY}

\begin{definition}[$\RDETL$ Theory]
  The \emph{axiomatic theory of Restricted Dynamic Epistemic Temporal
    Logic}, $\RDETL$, is defined in Figure~\ref{figure:RDETL}.  We
  write $\vdash_{\RDETL}\varphi$ to mean that $\varphi$ is a theorem
  of $\RDETL$.
\end{definition}

\begin{figure}[h]
  \begin{center}
    \textsc{The Axiomatic Theory $\RDETL$}\\[.2em]
    \renewcommand{\arraystretch}{1.3}
    \begin{tabular}[t]{rl}
      Restricted $\DETL$ Theory: &
      Schemes and rules from $\DETL$ restricted to $\lrdetl$\\
      
      Persistence of Facts: &
      $[Y]p\iff (\lnot[Y]\bot\imp p)$ \\
      
      Uniqueness of the Past: &
      $\lnot [Y]\varphi\imp [Y]\lnot\varphi$ \\

      Perfect Recall: &
      $[Y]\Box_a\varphi\imp\Box_a[Y]\varphi$ \\

      Knowledge of the Past: &
      $\lnot[Y]\bot\imp \Box_a\lnot[Y]\bot$ \\

      Knowledge of the Initial Time: &
      $[Y]\bot\imp \Box_a[Y]\bot$ \\
    \end{tabular}
  \end{center}
  \caption{The theory $\RDETL$; formulas and actions all come from
    $\lrdetl$}
  \label{figure:RDETL}
\end{figure}

\begin{theorem}[Soundness and Completeness for $\RDETL$]
\label{theorem:rdetl-completeness}
$\vdash_{\RDETL}\varphi$ iff $\models_{\RDETL}\varphi$ for each
$\varphi\in\lrdetl$.
\end{theorem}
\begin{proof}
  Theorem~\ref{theorem:soundness-completeness} already establishes the
  soundness of the $\DETL$ schemes and rules.  Soundness for the
  remaining schemes is straightforward to prove.

  The completeness proof can be divided into two stages.  First, prove
  the Reduction Theorem: every $L_\RDETL$-formula is $\RDETL$-provably
  equivalent to an action model-free formula in $\lsetl$.  This
  follows by the proof of Theorem~\ref{theorem:ldetl-reduction}.
  Second, prove completeness of action model-free formulas with
  respect to the class of restricted Kripke models:
  $\nvdash_\RDETL\psi$ for a given $\psi\in\lsetl$ implies there is a
  restricted situation $(M,w)$ for which $M,w\not\models\psi$.  We
  outline a proof of the second stage.

  To begin, fix $\psi\in\lsetl$ satisfying $\nvdash_\RDETL\psi$.  For
  each $\lsetl$-formula $\chi$, we define the \emph{$[Y]$-nesting
    depth} $d_Y(\chi)$ of $\chi$ by the following induction on the
  construction of $\chi$:
  \begin{eqnarray*}
    d_Y(\bot) &\eqdef& 0 \\
    d_Y(p) &\eqdef& 0 \\
    d_Y(\lnot\varphi) &\eqdef& d_Y(\varphi) \\
    d_Y(\varphi\land\psi) &\eqdef& \max\{d_Y(\varphi),d_Y(\psi)\} \\
    d_Y(\Box_a\varphi) &\eqdef& d_Y(\varphi) \\
    d_Y([Y]\varphi) &\eqdef& 1+d_Y(\varphi)
  \end{eqnarray*}
  Let $m\eqdef d_Y(\lnot\psi)$ be the $[Y]$-nesting depth of
  $\lnot\psi$.  We construct the canonical model
  $\Omega=(W^\Omega,\to^\Omega,\leadsfrom^\Omega,V^\Omega)$ for the
  theory $\RDETL$: the worlds in $W^\Omega$ are the maximal
  $\RDETL$-consistent sets of formulas, the binary relations are
  defined canonically (i.e., $\Gamma\to^\Omega_a\Delta$ iff
  $\Gamma^a\eqdef\{\chi\mid \Box_a\chi\in\Gamma\}\subseteq\Delta$ and
  $\Gamma\leadsfrom^\Omega\Delta$ iff $\Gamma^Y\eqdef\{\chi\mid
  [Y]\chi\in\Gamma\}\subseteq\Delta$), and the valuation is defined
  canonically as well (i.e., $V^\Omega(p)\eqdef\{\Gamma\in
  W^\Omega\mid p\in\Gamma\}$).  The Truth Lemma (i.e., the statement
  that $\varphi\in\Gamma$ iff $\Omega,\Gamma\models\varphi$ for each
  $\varphi\in\lsetl$) is proved in the usual way, and hence we have
  $\Omega,\Gamma_{\lnot\psi}\not\models\psi$ for a world
  $\Gamma_{\lnot\psi}\in W^\Omega$ obtained by a Lindenbaum
  construction as a maximal $\RDETL$-consistent extension of the
  $\RDETL$-consistent set $\{\lnot\psi\}$.  However, we cannot
  guarantee that $(\Omega,\Gamma_{\lnot\psi})\in\restricted_*$.  So to
  complete the argument, we perform a sequence of stepwise
  truth-preserving transformations on the canonical model in order to
  construct a pointed model $(F,A)\in\restricted_*$ satisfying the
  property that $F,A\not\models\psi$.
  \vspace{1.5ex}  

  \textsc{Unraveling:} We define the Kripke model $\Omega\times
  \mathbb{Z}$ as the following unraveling of $\Omega$ in the temporal
  direction:
    \begin{itemize}
    \item $W^{\Omega\times \mathbb{Z}}\eqdef W^\Omega\times \mathbb{Z}$.
      
    \item For each $a\in\Agnt$: $(w,k)\to^{\Omega\times \mathbb{Z}}_a
      (w',k')$ if and only if $w\to_a^\Omega w'$ and $k=k'$.

    \item $(w,k)\leadsfrom^{\Omega\times \mathbb{Z}} (w',k')$ if and
      only if $w \leadsfrom^\Omega w'$ and $k'=k-1$.

    \item $V^{\Omega\times \mathbb{Z}}(p) \eqdef V^\Omega(p) \times
      \mathbb{Z}$.

    \end{itemize}
    By induction on the construction of $\lsetl$-formulas, we have
    $\Omega\times \mathbb{Z}, (w,k)\models \varphi$ if and only if
    $\Omega, w\models \varphi$ for each $k\in\mathbb{Z}$ and each
    $\varphi\in \lsetl$.
    \vspace{1.5ex}

    \textsc{Generated submodel:} Let $M$ be the model generated by the
    world $(\Gamma_{\lnot\psi},m)$ in $\Omega\times\mathbb{Z}$ using
    the relations $\leadsfrom^{\Omega\times\mathbb{Z}}$ and
    $\to^{\Omega\times\mathbb{Z}}_a$ for each $a\in\Agnt$:
    \begin{itemize}
    \item $W^M\eqdef\{(w,k)\in\Omega\times\mathbb{Z}\mid
      (\Gamma_{\lnot\psi},m)(\leadsfrom^{\Omega\times\mathbb{Z}}\cup
      \bigcup_{a\in\Agnt}\to^{\Omega\times\mathbb{Z}}_a)^*(w,k)\}$,
      where $R^*$ denotes the reflexive-transitive closure of a binary
      relation $R$.

    \item For each $a\in\Agnt$: ${\to_a^M}\eqdef
      {\to_a^{\Omega\times\mathbb{Z}}}\cap(W^M\times W^M)$.

    \item
      ${\leadsfrom^M}\eqdef{\leadsfrom^{\Omega\times\mathbb{Z}}}\cap(W^M\times
      W^M)$.

    \item $V^M(p)\eqdef V^{\Omega\times\mathbb{Z}}(p)\cap W^M$.
    \end{itemize}
    By induction on the construction of $\lsetl$-formulas, we have
    $M,(w,k)\models\varphi$ if and only if
    $\Omega\times\mathbb{Z},(w,k)\models\varphi$ for each $(w,k)\in
    W^M$ and $\varphi\in\lsetl$.
    \vspace{1.5ex}

    \textsc{Trimming:} For each $\lsetl$-formula $\varphi$, we let
    $\Sub(\varphi)$ be the set of all subformulas of
    $\varphi$. If $S$ and $S'$ are sets of $\lsetl$-formulas, we
    define the following sets of $\lsetl$-formulas:
    \begin{eqnarray*}
      \Sub(S) &\eqdef& \textstyle\bigcup_{\chi\in S}\Sub(\chi)\\
      \lnot S &\eqdef& \{\lnot\chi\mid \chi\in S\} \\
      \Box S &\eqdef& \textstyle\bigcup_{a\in\Agnt}\{\Box_a\chi\mid\chi\in S\} \\{}
      [Y]S &\eqdef& \{[Y]\chi\mid\chi\in S\} \\
      S\imp S' &\eqdef& \{\chi\imp\chi'\mid \chi\in S\land \chi'\in S'\}\\
      S\iff S' &\eqdef& \{\chi\iff\chi'\mid \chi\in S\land \chi'\in S'\}
    \end{eqnarray*}
    For each $k\in\mathbb{N}$ and $\lsetl$-formula $\varphi$, we let
    $\Sub_k(\varphi)$ be the set of subformulas of $\varphi$ with
    $[Y]$-nesting depth at most $k$ and we define the set of
    $\lsetl$-formulas $\Cl_k(\varphi)$ by the following induction:
    \begin{eqnarray*}
      \Cl_0(\varphi) &\eqdef& 
      \begin{array}[t]{ll}
        \textstyle\Sub_0(\varphi)\cup\{\bot\}
      \end{array}
      \\
      \Cl_{k+1}(\varphi) &\eqdef&
      \renewcommand{\arraystretch}{1.3}
      \begin{array}[t]{ll}
        \Sub_{k+1}(\varphi) & \cup
        \\{}
        \Sub([Y]\Cl_k(\varphi) \iff
        (\lnot[Y]\Cl_k(\varphi)\imp\Cl_k(\varphi)))
        & \cup
        \\{}
        \Sub(\lnot[Y]\Cl_k(\varphi)\imp
        [Y]\lnot\Cl_k(\varphi))
        & \cup
        \\{}
        \Sub(\lnot[Y]\Cl_k(\varphi)\imp
        \Box\lnot[Y]\Cl_k(\varphi))
        & \cup
        \\{}
        \Sub([Y]\Cl_k(\varphi)\imp
        \Box[Y]\Cl_k(\varphi))
      \end{array}
    \end{eqnarray*}
    Observe that $\Cl_k(\varphi)$ contains $\lsetl$-formulas
    of $[Y]$-nesting depth at most $k$.  
    Note also that $\bot\in \Cl_k(\varphi)$ for each $k$.
    We define a Kripke model $M'$
    consisting of all worlds $(w,k)\in W^M$ satisfying $0\leq k\leq m$
    with all other components relativized to this set of worlds:
    \begin{itemize}
    \item $W^{M'}\eqdef\{(w,k)\in W^M\mid 0\leq k\leq m\}$.

    \item For each $a\in\Agnt$: ${\to_a^{M'}}\eqdef
      {\to_a^M}\cap(W^{M'}\times W^{M'})$.

    \item ${\leadsfrom^{M'}}\eqdef{\leadsfrom^M}\cap(W^{M'}\times
      W^{M'})$.

    \item $V^{M'}(p)\eqdef V^M(p)\cap W^{M'}$.
    \end{itemize}
    By an induction on $k$ satisfying $0\leq k\leq m$ with a
    subinduction on the construction of $\lsetl$-formulas of
    $[Y]$-nesting depth at most $k$, we can show that for each
    $\varphi\in\Cl_k(\lnot\psi)$ and $(w,k)\in W^{M'}$, we have
    $M',(w,k)\models\varphi$ if and only if $M,(w,k)\models\varphi$.
    \vspace{1.5ex}

    \textsc{Filtration:} We define an equivalence relation $\equiv$ on
    elements of $W^{M'}$ by setting $(w,k)\equiv (w',k')$ if and only
    if $k=k'$ and for all $\varphi\in \Cl_k(\lnot\psi)$, we
    have that $M',(w,k)\models\varphi$ iff $M',(w',k')\models\varphi$.
    For a world $(w,k)\in W^{M'}$, we write $[w,k]$ to denote the
    equivalence class
    \[
    [w,k]\eqdef \{(w',k')\in W^{M'}\mid (w',k')\equiv (w,k)\}
    \]
    of $(w,k)$ under $\equiv$.  We define a Kripke model $F$ by the
    equivalence relation $\equiv$ as follows:
    \begin{itemize}
    \item $W^F\eqdef \{[w,k]\mid (w,k)\in W^{M'}\}$.

    \item For each $a\in\Agnt$: $A\to^F_a B$ iff $\exists (w,k)\in A$
      and $\exists(v,j)\in B$ with $(w,k)\to^{M'}_a(v,j)$.

    \item $A\leadsfrom^F B$ iff $\exists(w,k)\in A$ and $\exists(v,j)\in B$
      with $(w,k)\leadsfrom^{M'}(v,j)$.

    \item $V^F(p)\eqdef \{A\in W^F\mid p\in\Cl_m(\lnot\psi) \text{
        and } \forall(w,k)\in A:M',(w,k)\models p\}$.
    \end{itemize}
    By induction on $k$ satisfying $0\leq k\leq m$ with a subinduction
    on the construction of $\lsetl$-formulas of $[Y]$-nesting depth at
    most $k$, we can show that for each
    $\varphi\in\Cl_k(\lnot\psi)$ and $[w,k]\in W^F$, we have
    $F,[w,k]\models\varphi$ if and only if $M',(w',k')\models\varphi$
    for each $(w',k')\in[w,k]$.
    \vspace{1.7ex}

    \textsc{Truth preservation:} By what we have shown above, it
    follows that for each $(w,k)\in W^\Omega\times\{0,\dots,m\}$ and
    each $\varphi\in\Cl_k(\lnot\psi)$, we have
    $\Omega,w\models\varphi$ if and only if
    $F,[w,k]\models\varphi$.  In particular, we have
    $F,[\Gamma_{\lnot\psi},m]\not\models\psi$.  So to complete the
    proof, it suffices for us to show that $F\in\restricted$ (i.e.,
    $F$ is a restricted model).
    \vspace{1.7ex}

    \textsc{$F$ is a restricted model:} Before we proceed, note that
    $(x,k)\to_a^{M'} (x',k')$ implies $k=k'$ and that
    $(x,k)\leadsto^{M'} (x',k')$ implies $k=k'-1$.  Thus if $A \to^F_a
    B$, then all pairs in $A$ and $B$ have the same second coordinate.
    Also, if $A\leadsto^F B$, then all pairs in $A$ have a second
    coordinate one less than the second coordinate of the pairs in
    $B$.
    \begin{itemize}
    \item \emph{$F$ satisfies uniqueness of the past:}
      \begin{center}
        \footnotesize $([w',k']\leadsto^F[w,k] \land
        [w'',k'']\leadsto^F[w,k])\Rightarrow([w',k']=[w'',k''])
        \enspace.$
      \end{center}

      Suppose not.  From $[w',k']\leadsto^F[w,k]$ and
      $[w'',k'']\leadsto^F[w,k]$, we have $k'=k''=k-1$.  From
      $[w',k-1]\neq[w'',k-1]$, it follows that there exists
      $\varphi\in\Cl_{k-1}(\lnot\psi)$ such that, without loss
      of generality, $F,[w',k-1]\models\varphi$ and
      $F,[w'',k-1]\not\models\varphi$.  Defining
      \[
      \chi\eqdef \lnot[Y]\varphi\imp[Y]\lnot\varphi\enspace,
      \]
      we have $\chi\in\Cl_k(\lnot\psi)$ and
      $F,[w,k]\not\models\chi$ and hence that
      $\Omega,w\not\models\chi$ by \textsc{Truth preservation}.  But
      then it follows by the Truth Lemma that the maximal consistent
      set $w$ fails to contain an instance $\chi$ of the Uniqueness of
      the Past axiom, a contradiction.

    \item \emph{$F$ satisfies persistence of facts:}
      \begin{center}
        \footnotesize $[w,k]\leadsto^F[v,j] \Rightarrow ([w,k]\in
        V^F(p)\Leftrightarrow[v,j]\in V^F(p))\enspace.$
      \end{center}

      Suppose not.  Then we have
      \begin{equation}
        \lnot([w,k]\in V^F(p)\Leftrightarrow[v,j]\in V^F(p))\enspace.
        \label{eq:RDETL:p}
      \end{equation}
      Let $\chi\eqdef [Y]p\iff(\lnot[Y]\bot\imp p)$. Since $F$
      satisfies uniqueness of the past, it follows from
      \eqref{eq:RDETL:p} that $F,[v,j]\not\models \chi$.  Since
      $[w,k]\leadsto^F[v,j]$, we have $j\geq 1$.  Further, it follows
      by \eqref{eq:RDETL:p} that $[w,k]\in V^F(p)$ or $[v,j]\in
      V^F(p)$.  Applying the definition of $V^F(p)$, we have that
      $p\in\Cl_m(\lnot\psi)$ and therefore that
      $p\in\Cl_0(\lnot\psi)$. So since $j\geq 1$, we have
      $\chi\in\Cl_j(\lnot\psi)$.  But then it follows from
      $F,[v,j]\not\models \chi$ by \textsc{Truth preservation} that
      $\Omega,v\not\models\chi$.  Applying the Truth Lemma, the
      maximal consistent set $v$ fails to contain an instance $\chi$
      of the Persistence of Facts axiom, a contradiction.

    \item \emph{$F$ is depth-defined:} $\mydepth([w,k])\neq\infty$.

      For each $[v,j]\in W^F$, we have $0\leq j\leq m$.  Further,
      $[w',k']\leadsfrom^F[v',j']$ implies $j'=k'-1$.  It follows that
      $\mydepth([w,k])\neq\infty$.

    \item \emph{$F$ satisfies knowledge of the past:}
      \begin{center}
        \footnotesize $([w',k']\leadsto^F[w,k]\to^F_a[v,j])\Rightarrow
        \exists[v',j']([v',j']\leadsto^F[v,j])\enspace.$
      \end{center}

      Suppose not.  Letting
      $\chi\eqdef\lnot[Y]\bot\imp\Box_a\lnot[Y]\bot$, it follows that
      $F,[w,k]\not\models\chi$.  Further, from
      $[w',k']\leadsto^F[w,k]$, we have $k\geq 1$ and therefore that
      $\chi\in\Cl_k(\lnot\psi)$.  But then it follows by
      \textsc{Truth preservation} that $\Omega,w\not\models\chi$.
      Applying the Truth Lemma, the maximal consistent set $w$ fails
      to contain an instance $\chi$ of the Uniqueness of the Past
      axiom, a contradiction.

    \item \emph{$F$ satisfies knowledge of the initial time:}
      \begin{center}
        \footnotesize $[w,k]\to^F_a[v,j]\land
        \lnot\exists[w',k']([w',k']\leadsto^F[w,k])\Rightarrow
        \lnot\exists[v',j']([v',j']\leadsto^F[v,j])\enspace.$
      \end{center}
      
      Suppose $[w,k]\to^F_a[v,j]$ and
      $\lnot\exists[w',k']([w',k']\leadsto^F[w,k])$.  If $k=0$, then
      it follows by $[w,k]\to^F_a[v,j]$ that $j=0$ and therefore that
      $\lnot\exists[v',j']([v',j']\leadsto^F[v,j])$ by
      \textsc{Trimming}.  So let us assume that $k>0$.  Further,
      toward a contradiction, we assume that
      $\exists[v',j']([v',j']\leadsto^F[v,j])$.  Letting
      $\chi\eqdef[Y]\bot\imp\Box_a[Y]\bot$, it follows that
      $F,[w,k]\not\models\chi$.  But since $k>0$, we have
      $\chi\in\Cl_k(\lnot\psi)$ and hence it follows by
      \textsc{Truth preservation} that $\Omega,w\not\models\chi$.
      Applying the Truth Lemma, the maximal consistent set $w$ fails
      to contain an instance $\chi$ of the Knowledge of the Initial
      Time axiom, a contradiction.

    \item \emph{$F$ satisfies perfect recall:}
      \begin{center}
        \footnotesize
        $([w,k]\leadsto^F[v,j]\to^F_a[v',j'])\Rightarrow
        \exists[w',k']([w,k]\to^F_a[w',k']\leadsto^F[v',j'])
        \enspace.$
      \end{center}

      Suppose $[w,k]\leadsto^F[v,j]\to^F_a[v',j']$.  Then $k+1=j=j'$.
      Now $[v,k+1]=[v,j]\to^F_a[v',j']=[v',k+1]$ implies
      $\exists(v_*,k+1)\in[v,k+1]$ and $\exists(v'_*,k+1)\in[v',k+1]$
      with $(v_*,k+1)\to^{M'}_a(v'_*,k+1)$.  Hence
      $v_*\to^\Omega_av'_*$.  Now from $[w,k]\leadsto^F[v,j]=[v,k+1]$,
      we have that $F,[v,k+1]\models\lnot[Y]\bot$.  Since $k+1\geq 1$,
      it follows that $\lnot[Y]\bot\in\Cl_{k+1}(\lnot\psi)$
      and therefore we have by \textsc{Truth preservation} that
      $\Omega,v_*\models\lnot[Y]\bot$.  By the definition of truth, it
      follows that there is a $w_*\in W^\Omega$ with
      $w_*\leadsto^\Omega v_*$.  But then
      $[w_*,k]\leadsto^F[v,j]=[v,k+1]$, from which it follows by
      uniqueness of the past for $F$ that $[w_*,k]=[w,k]$ and
      therefore that $(w_*,k)\in[w,k]$. Now $\Omega$ satisfies perfect
      recall (for if it did not, we could find a violation of an
      instance of the Perfect Recall axiom at a maximal consistent
      set, a contradiction). 
      Therefore, since we have
      $w_*\leadsto^\Omega v_*\to^\Omega_a v'_*$, it follows by perfect
      recall for $\Omega$ that there is a $w'_*\in W^\Omega$
      satisfying $w_*\to^\Omega_a w'_*\leadsto^\Omega v'_*$.  
      But then
      $[w,k]=[w_*,k]\to^\Omega_a[w'_*,k]\leadsto^\Omega[v'_*,k+1]=[v',k+1]=[v',j]$.
      \qedhere
    \end{itemize}
  \end{proof}

\subsection{$\YDEL$}

In this section, we relate $\DETL$ to a more conservative approach to
adding time to $\DEL$, which, as is generally studied in the
literature
\cite{BalMos04,BalMosSol98,BalDitMos08,BenEijKoo06,DitHoeKoo07}, does
not use $\leadsto$ arrows in its action models. Further, the semantics
of $\DEL$ does not use Kripke models with designated time-keeping
arrows $\leadsto$. In order to draw this comparison, we will define an
extension of $\DEL$ called ``Yesterday Dynamic Epistemic Logic,'' or
$\YDEL$ (see \cite{Sac08,Sac10,Yap07}), that records a history of the
updates made. We will then show that $\YDEL$ reasoning can be done
within the $\DETL$ setting, since (modulo translation) $\YDEL$ is
sound and complete with respect to a particular class of $\DETL$
models.

\begin{definition}
  \label{definition:lydel}
  $\lydel$ is the atemporal fragment of $\ldetl$.  For reasons
  explained in a moment, we assume that the special symbol $\flat$ is
  used neither as a world nor as an event in $\lydel$.
\end{definition}

We will evaluate $\YDEL$ formulas on restricted Kripke models
(Definition~\ref{def:restricted}).

\begin{definition}[$\YDEL$ Semantics]
  \label{def:YDEL-semantics}
  We define a relation $\models_\YDEL$ between pointed models in
  $\restricted_\pt$ and $\lydel$-formulas using standard Boolean cases
  and the following modal cases.
  \[
  \renewcommand{\arraystretch}{1.3}
  \begin{array}{lcl}
    M,w \models_\YDEL \Box_a\varphi &\text{iff}&
    M,v\models_\YDEL\varphi
    \text{ for each }
    v\from^M_a w
    \\
    M,w \models_\YDEL [Y]\varphi &\text{iff}&
    M,v\models_\YDEL\varphi
    \text{ for each }
    v\leadsto^M w
    \\
    M,w \models_\YDEL [U,s]\varphi &\text{iff}&
    M,w\models_\YDEL\pre^U(s) \text{ implies }
    M\oplus U,(w,s)\models_\YDEL\varphi
  \end{array}
  \]
  where
  \[
  \renewcommand{\arraystretch}{1.3}
  \begin{array}{rcl}
    W^{M\oplus U}
    &\eqdef&
    (W^M\times\{\flat\}) \;\cup
    \\
    &&
    \{(v,t)\in W^M\times W^U\mid M,v\models_\YDEL\pre^U(t)\}
    \\
    (v,t)\to^{M\oplus U}_a(v',t')
    &\text{iff}&
    ((t,t'\neq\flat)\;\&\; v\to^M_a v' \;\&\; t\to^U_a t') \text{ or }
    \\
    &&
    ((t=t'=\flat) \;\&\; v\to^M_a v')
    \\
    (v,t)\leadsto^{M\oplus U}(v',t')
    &\text{iff}&
    (t=\flat\;\&\; t'\neq\flat \;\&\; v=v' )
    \text{ or }
    \\
    &&
    ((t=t'=\flat) \;\&\; v\leadsto^M v')
    \\
    (v,t)\in V^{M\oplus U}(p) &\text{iff}&
    v\in V^M(p)
  \end{array}
  \]
\end{definition}

The forthcoming Corollary~\ref{corollary:ydel-closure} shows that
$M\oplus U\in\restricted$ whenever $M\in \restricted$ and $U$ is
atemporal.  The function of the symbol $\flat$ is to serve as an
epistemic past state, preserving a copy of $M$ in $M \oplus U$ (Lemma
\ref{lemma:flat-bisim}). Since $\YDEL$ uses atemporal action models,
which contain no $\leadsto$ arrows, the mechanism for preserving the
previous model $M$ is ``hardcoded'' in the semantics. However, by
defining a translation from $\lydel$ to $\ldetl$ (Definition
\ref{definition:sharp}), we will be able to show that ``atemporal''
$\YDEL$ reasoning can be captured in $\DETL$.

\begin{lemma}
  \label{lemma:flat-bisim}
  Let $(M,w)$ be a situation and $(U,s)$ be an atemporal action
  satisfying $M,w\models\pre^U(s)$.  The function $f:W^M\to W^{M\oplus
    U}$ defined by $f(w)\eqdef(w,\flat)$ is a bisimulation.
\end{lemma}

\begin{proof}
  $w$ and $(w,\flat)$ have the same valuation.  If $f(w)=(w,\flat)
  \to^{M \oplus U}_a (v,t)$, then $t=\flat$ and $w \to^M_a v$. If $w
  \to^M_a v$, then $(w, \flat) \to^{M \oplus U}_a (v,\flat)$.
\end{proof}

Before defining the translation from $\YDEL$ to $\DETL$, we will first
show how $\YDEL$ works by illustrating the way in which $M \oplus U$
is constructed.

\begin{example}
  Figure~\ref{figure:oplus-initial} pictures an initial situation and
  a $\YDEL$ action.  In the initial situation, neither agent knows
  whether $p$ is true.  The action informs $a$ that $p$ is true but
  tells $b$ only that $a$ was either informed of $p$ or provided with
  trivial information.  After applying the action to the situation, we
  obtain the resultant situation in Figure~\ref{figure:oplus}.

  \begin{figure}[ht]
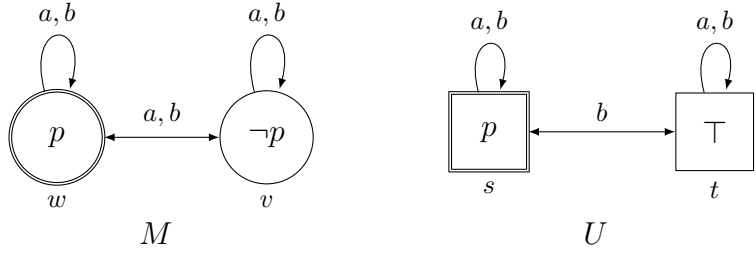

    \begin{center}
      \begin{tabular}{c@{\qquad\qquad}c}
        
        \begin{mytikz}
          \node[w,double,label={below:$w$}] (w) {$p$};
          
          \node[w,right of=w,label={below:$v$}] (v) {$\lnot p$};
                    
          \path (w) edge[loop above] node{$a,b$} ();
          
          \path (w) edge[<->] node{$a,b$} (v);

          \path (v) edge[loop above] node{$a,b$} ();
        \end{mytikz}
        &
        \begin{mytikz}
          
          \node[e,below of=w,double,label={below:$s$}] (s) {$p$};
                              
          \node[e,right of=s,label={below:$t$}] (t) {$\top$};
          
          \path (s) edge[loop above] node{$a,b$} ();
          
          \path (s) edge[<->] node{$b$} (t);
          
          \path (t) edge[loop above] node{$a,b$} ();
        \end{mytikz}
        \\
        $M$ &
        $U$
      \end{tabular}
    \end{center}
    \caption{A situation $(M,w)$ and a $\YDEL$ action model $(U,s)$.}
    \label{figure:oplus-initial}
  \end{figure}
  
  \begin{figure}
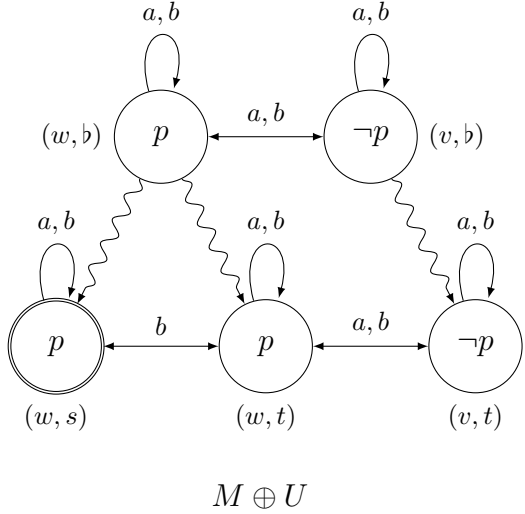

    \begin{center}
      \begin{tabular}{c}
        \begin{mytikz}
          \node[w,label={left:$(w,\flat)$}] (wb) {$p$};
          
          \node[w,right of=wb,label={right:$(v,\flat)$}] (vb) {$\lnot p$};
          
          \node[w,below of=wb,xshift=3.375em,label={below:$(w,t)$}] (wt) {$p$};

          \node[w,right of=wt,label={below:$(v,t)$}] (vt) {$\lnot p$};

          \node[w,double,left of=wt,label={below:$(w,s)$}] (ws) {$p$};
                                        
          \path (wb) edge[loop above] node{$a,b$} ();
          
          \path (vb) edge[loop above] node{$a,b$} ();
          
          \path (ws) edge[loop above] node{$a,b$} ();
          
          \path (wt) edge[loop above] node{$a,b$} ();
          
          \path (vt) edge[loop above] node{$a,b$} ();
          
          \path (wb) edge[<->] node{$a,b$} (vb);
          
          \path (ws) edge[<->] node{$b$} (wt);
          
          \path (wt) edge[<->] node{$a,b$} (vt);
                    
          \path (wb) edge[zz,->] node{} (ws);
          
          \path (wb) edge[zz,->] node{} (wt);
          
          \path (vb) edge[zz,->] node{} (vt);
        \end{mytikz}
        \\\\
        $M\oplus U$
      \end{tabular}
    \end{center}
    \caption{The situation $(M \oplus U, (w,s))$ resulting from
      application of $(U,s)$ to $(M,w)$, both from
      Figure~\ref{figure:oplus-initial}. Agent arrows $\to_x$ are here
      implicitly closed under transitivity.}
    \label{figure:oplus}
  \end{figure}
  
\end{example}

We now show how $\YDEL$ reasoning is captured in $\DETL$.

\subsubsection{Translation of $L_\YDEL$ into $L_\DETL$}

We define a translation from $\YDEL$ formulas and action models to
$\DETL$ formulas and action models.  This translation acts on action
models by adding a new epistemic past state $\flat$ along with an
arrow $\flat\leadsto s$ to each action $s$.  See
Figure~\ref{figure:sharpPicture} for an example.

\begin{definition}[$\sharp$ Translation]
\label{definition:sharp}
  We define a function
  \[
  \sharp:\lydel\cup\Actm^a(\lydel)\to \ldetl\cup\Actm(\ldetl)
  \]
  as follows: $\bot^\sharp=\bot$, $p^\sharp=p$, $\sharp$ commutes with
  unary Boolean connectives and with non-$[U,s]$ modal connectives,
  $\sharp$ distributes over binary Boolean connectives,
  $([U,s]\varphi)^\sharp=[U^\sharp,s]\varphi^\sharp$, and $U^\sharp$
  is defined by taking
  \[
  \renewcommand{\arraystretch}{1.3}
  \begin{array}{lcl}
    W^{U^\sharp} &\eqdef& W^U \cup \{\flat\}
    \\
    s\to^{U^\sharp}_a s' & \text{iff} &
    s\to^U_a s' \text{ or } s=s'=\flat
    \\
    s\leadsto^{U^\sharp} s' & \text{iff} & s=\flat \text{ and } s'\in W^U
    \\
    \pre^{U^\sharp}(s) & \eqdef & \pre^U(s)^\sharp \text{ for } s\in W^U \\
    \pre^{U^\sharp}(\flat) & \eqdef & \top
  \end{array}
  \]

\end{definition}

\begin{figure}
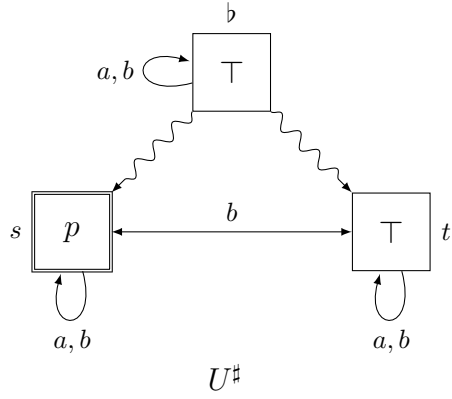

  \begin{center}
    \begin{tabular}{cc}
      \begin{mytikz}
                
        \node[e,label={above:$\flat$}] (flat) {$\top$};
        
        \node[e,double,below left of=flat,label={left:$s$}] (s) {$p$};
        
        \node[e,below right of=flat,label={right:$t$}] (t) {$\top$};
        
        \path (s) edge[loop below] node{$a,b$} ();
        
        \path (s) edge[<->] node{$b$} (t);
        
        \path (t) edge[loop below] node{$a,b$} ();
        
        \path (flat) edge[loop left] node{$a,b$} ();
        
        \path (flat) edge[zz,->] node{} (s);
        
        \path (flat) edge[zz,->] node{} (t);
      \end{mytikz}
      \\
      $U^\sharp$
    \end{tabular}
  \end{center}
  \caption{The action $(U^\sharp,s)$ obtained by applying the
    translation $\sharp$ to the $\YDEL$ action model $U$ from
    Figure~\ref{figure:oplus-initial}.}
  \label{figure:sharpPicture}
\end{figure}

The function $\sharp$ transforms the atemporal action models used by
$\YDEL$ into $\DETL$ action models having epistemic past states.  As
it turns out, such action models are in fact $\RDETL$ action models
(Definition~\ref{definition:lrdetl}).

\begin{lemma}
  \label{lemma:usharp-rdetl}\label{lemma:usharp-properties}
  If $U$ is an atemporal action model, then
  $U^\sharp\in\Actm(\lrdetl)$.
\end{lemma}

The proof of this lemma is straightforward.  It follows that the image
of $\sharp$ is contained in $\lrdetl\cup\Actm(\lrdetl)$.  This
containment is strict: every history in $U^\sharp$ has length $1$,
while the length of histories in $\RDETL$ action models is unbounded.

The situation $(M\oplus U,(w,s))$ from Figure~\ref{figure:oplus} was
produced by applying the $\YDEL$ action $(U,s)$ to the initial
situation $(M,w)$, both from Figure~\ref{figure:oplus-initial}.  It is
not difficult to verify that we obtain the same final situation by
applying the $\RDETL$ action $(U^\sharp,s)$; that is,
\[
(M\oplus U,(w,s)) = (M[U^\sharp],(w,s))\enspace.
\]
The following theorem shows that this result holds in general.

\begin{theorem}
  \label{theorem:ydel-rdetl}
  For each $(M,w)\in\restricted_*$, each $\varphi\in\lydel$, and each
  $U\in\Actm^a(\lydel)$:
  \begin{enumerate}
     \renewcommand{\theenumi}{(\alph{enumi})} 
    \renewcommand{\labelenumi}{\theenumi}
  \item $M\oplus U=M[U^\sharp]$, and
    
  \item $M,w\models_\YDEL\varphi$ iff $M,w\models_\RDETL\varphi^\sharp$.
  \end{enumerate}
\end{theorem}
\begin{proof}
  Set $L_0\eqdef\lsetl$.  Once $L_i$ is defined, define the language
  $L_{i+1}$ and the set $\Actm^a_*(L_{i+1})$ of pointed atemporal
  action models over $L_{i+1}$ by the following grammar:
  \[
  \begin{array}{lcl}
    \varphi &::=&
    \psi \mid \lnot\varphi \mid \varphi\land\varphi \mid
    \Box_a\varphi \mid [Y]\varphi
    \mid [U,s]\varphi
    \\
    &&
    \text{\footnotesize
      $\psi\in L_i$,
      $a\in\Agnt$,
      $(U,s)\in\Actm^a_\pt(L_i)$
    }
  \end{array}
  \]
  Note that the preconditions of any $U\in\Actm^a(L_{i+1})$ all come
  from $L_i$.  Further, $\lydel=\bigcup_{i\in\Nat}L_i$.  We prove by
  induction on $i$ that we have $S_i$, which we define as the
  conjunction of the following two statements:
  \begin{enumerate}
  \item \label{ydel:actm} For each $M\in\restricted$ and each
    $U\in\Actm^a(L_i)$: $M\oplus U=M[U^\sharp]$.

  \item \label{ydel:flma} For each $(M,w)\in\restricted_*$ and each
    $\varphi\in L_i$:
    \begin{center}
      $M,w\models_\YDEL\varphi$ iff
      $M,w\models_\RDETL\varphi^\sharp$\enspace.
    \end{center}
  \end{enumerate}
  The base case $i=0$ is immediate.  For the induction step, we assume
  that $S_i$ holds (the ``induction hypothesis'') and prove that
  $S_{i+1}$ holds. We begin with Statement~\ref{ydel:actm}. We must
  show that $W^{M\oplus U}=W^{M[U^\sharp]}$, that ${\to_a^{M\oplus
      U}}={\to_a^{M[U^\sharp]}}$, that ${\leadsfrom^{M\oplus
      U}}={\leadsfrom^{M[U^\sharp]}}$, and that $V^{M\oplus
    U}=V^{M[U^\sharp]}$.  Proceeding, we have
  \begin{eqnarray*}
    W^{M\oplus U} &=&
    (W^M\times\{\flat\}) \cup
    \{(v,t)\in W^M\times W^U\mid M,v\models_\YDEL\pre^U(t)\}
    \\
    &=&
    (W^M\times\{\flat\}) \cup
    \{(v,t)\in W^M\times W^U\mid M,v\models_\RDETL\pre^U(t)^\sharp\}
    \\
    &=&
    \{(v,t)\in W^M\times W^{U^\sharp}\mid
    M,v\models_\RDETL\pre^U(t)^\sharp\}
    \\
    &=&
    \{(v,t)\in W^M\times W^{U^\sharp}\mid
    M,v\models_\RDETL\pre^{U^\sharp}(t)\}
    \\
    &=&
    W^{U^\sharp}
  \end{eqnarray*}
  The first equality follows by definition of $W^{M\oplus U}$.  The
  second follows by the induction hypothesis.  The third follows by
  the definition of $U^\sharp$; in particular, $\flat\in W^{U^\sharp}$
  has precondition $\top$. The fourth equality follows by the
  definition of $U^\sharp$.  The fifth equality follows by the
  definition of $W^{M[U^\sharp]}$.

  Next, we have $(v,t)\to^{M\oplus U}_a(v',t')$ by definition if and
  only if $v\to^M_av'$ and either
  \begin{itemize}
  \item $t,t'\neq\flat$ and $t\to^U_at'$; or  

  \item $t=t'=\flat$.
  \end{itemize}
  Further, we have $(v,t)\to^{M[U^\sharp]}_a(v',t')$ by definition if
  and only if $v\to^M_av'$ and $t\to^{U^\sharp}_at'$.  Thus
  ``$(v,t)\to^{M[U^\sharp]}_a(v',t')$ if and only if
  $(v,t)\to^{M\oplus U}_a(v',t')$'' (which we aim to show) is
  equivalent to ``$t\to^{U^\sharp}_at'$ if and only if $t=t'=\flat$ or
  both $t,t'\neq\flat$ and $t\to^U_at'$.''  But this follows because
  $\flat\notin W^U$ by assumption (Definition~\ref{definition:lydel})
  and we have $t\to^{U^\sharp}_at'$ if and only if $t\to^U_at'$ or
  $t=t'=\flat$ (Definition~\ref{definition:sharp}).  Conclusion:
  ${\to_a^{M\oplus U}}={\to_a^{M[U^\sharp]}}$.

  Next, we have $(v,t)\leadsto^{M\oplus U}(v',t')$ by definition if
  and only if
  \begin{itemize}
  \item $t=\flat$, $t'\neq\flat$, and $v=v'$; or

  \item $t=t'=\flat$ and $v\leadsto^M v'$.
  \end{itemize}
  Further, we have $(v,t)\leadsto^{M[U^\sharp]}(v',t')$ by definition
  if and only if 
  \begin{itemize}
  \item $v\leadsto^Mv'$, $t=t'$, and $t$ is a past state; or

  \item $v=v'$ and $t\leadsto^{U^\sharp} t'$.
  \end{itemize}
  In the action model $U^\sharp$, event $\flat$ is the unique past
  state, the only event with an outgoing $\leadsto$ arrow, and has an
  outgoing $\leadsto$ arrow to every other event.  It follows that we
  have $(v,t)\leadsto^{M[U^\sharp]}(v',t')$ if and only if
  ``$v\leadsto^Mv'$ and $t=t'=\flat$'' or ``$v=v'$, $t=\flat$, and
  $t'\neq\flat$.''  But this is equivalent to the conditions defining
  $(v,t)\leadsto^{M\oplus U}(v',t')$.  Conclusion:
  ${\leadsfrom^{M\oplus U}}={\leadsfrom^{M[U^\sharp]}}$.

  Finally, we have $(v,t)\in V^{M\oplus U}(p)$ if and only if $v\in
  V^M(p)$ if and only if $(v,t)\in V^{M[U^\sharp]}(p)$.  Here we made
  tacit use of the fact that $W^{M\oplus U}=W^{M[U^\sharp]}$.
  Conclusion: $V^{M\oplus U}=V^{M[U^\sharp]}$.  

  This completes the proof of Statement~\ref{ydel:actm}.  The proof of
  Statement~\ref{ydel:flma} then proceeds by a sub-induction on the
  construction of $L_{i+1}$-formulas.  Most cases are obvious, so we
  only address the case for $L_{i+1}$-formulas $[U,s]\varphi$.
  Proceeding, we have $M,w\models_\YDEL[U,s]\varphi$ if and only if
  $M,w\not\models_\YDEL\pre^U(s)$ or $M\oplus
  U,(w,s)\models_\YDEL\varphi$.  By Statement~\ref{ydel:flma} of the
  induction hypothesis, the latter is equivalent to
  ``$M,w\not\models_\RDETL\pre^U(s)^\sharp$ or $M\oplus
  U,(w,s)\models_\RDETL\varphi^\sharp$,'' which is itself equivalent
  to ``$M,w\not\models_\RDETL\pre^{U^\sharp}(s)$ or
  $M[U^\sharp],(w,s)\models_\RDETL\varphi^\sharp$'' (by the definition
  of $U^\sharp$ for the left disjunct and Statement~\ref{ydel:actm} of
  the induction hypothesis for the right).  But this is equivalent to
  $M,w\models_\RDETL[U^\sharp,s]\varphi^\sharp$ by the $\RDETL$
  semantics.  Since
  $[U^\sharp,s]\varphi^\sharp=([U,s]\varphi)^\sharp$, the result
  follows.
\end{proof}

A corollary of Theorem~\ref{theorem:ydel-rdetl} is that restricted
models are closed under the $\YDEL$ update operation $M\mapsto M\oplus
U$.

\begin{corollary}
  \label{corollary:ydel-closure}
  If $M\in\restricted$, $U\in\Actm(\lydel)$, and $W^{M\oplus
    U}\neq\emptyset$, then $M\oplus U\in\restricted$.
\end{corollary}
\begin{proof}
  Fix $M\in\restricted$ and $U\in\Actm(\lydel)$ with $W^{M\oplus
    U}\neq\emptyset$.  Since $U\in\Actm(\lydel)$ if and only if $U$ is
  atemporal (Definition~\ref{definition:lydel}), it follows by
  Lemma~\ref{lemma:usharp-rdetl} that $U^\sharp\in\Actm(\lrdetl)$.
  Applying Preservation (Theorem~\ref{theorem:preservation}) and the
  definition of $\lrdetl$ (Definition~\ref{definition:lrdetl}), we
  have $M[U^\sharp]\in\restricted$.  By
  Theorem~\ref{theorem:ydel-rdetl}, $M\oplus
  U=M[U^\sharp]\in\restricted$.
\end{proof}

\subsubsection{Connecting the theories of $\YDEL$ and $\RDETL$}

\begin{definition}
  The theory $\YDEL$ is defined in Figure~\ref{figure:YDEL}.  Note that
  all axioms and rules refer to formulas in $\lydel$ and hence to
  action models in $\Actm(\lydel)$.
\end{definition}

\begin{figure}[h]
  \begin{center}
    \textsc{Basic Schemes}\\[.2em]
    \renewcommand{\arraystretch}{1.3}
    \begin{tabular}[t]{rl}
      Classical Logic: &
      Schemes for Classical Propositional Logic \\

      Normality for $\Box_a$: &
      $\Box_a(\varphi\imp\psi)\imp(\Box_a\varphi\imp\Box_a\psi)$ \\

      Normality for $Y$: &
      $[Y](\varphi\imp\psi)\imp([Y]\varphi\imp[Y]\psi)$ \\

      Persistence of Facts: &
      $[Y]p\iff (\lnot[Y]\bot\imp p)$ \\

      Uniqueness of the Past: &
      $\lnot [Y]\varphi\imp [Y]\lnot\varphi$ \\

      Perfect Recall: &
      $[Y]\Box_a\varphi\imp\Box_a[Y]\varphi$ \\

      Knowledge of the Past: &
      $\lnot[Y]\bot\imp \Box_a\lnot[Y]\bot$ \\

      Knowledge of the Initial Time: &
      $[Y]\bot\imp \Box_a[Y]\bot$ \\

      Reduction Axioms: &
      $[U,s]q \iff \bigl( \pre^U(s)\imp q \bigr)$
      for $q\in\Prop$ \\

      & $[U,s](\varphi\land\psi) \iff
      \bigl( [U,s]\varphi \land [U,s]\psi \bigr)$
      \\

      & $[U,s]\lnot\varphi \iff
      \bigl( \pre^U(s)\imp\lnot[U,s]\varphi \bigr)$ \\

      & $[U,s]\Box_a\varphi \iff
      \bigl(
      \pre^U(s) \imp
      \bigwedge_{s'\from^U_as}\Box_a[U,s']\varphi
      \bigr)$ \\

      & $[U,s][Y]\varphi \iff
      (\pre^U(s)\imp\varphi)$
    \end{tabular}
    \\[.7em]
    \textsc{Rules}\vspace{-1em}
    \[
    \begin{array}{c}
      \varphi\imp\psi \quad \varphi
      \\\hline
      \psi
    \end{array}
    \,\text{\footnotesize (MP)}
    \quad
    \begin{array}{c}
      \varphi
      \\\hline
       \Box_a\varphi
    \end{array}
    \,\text{\footnotesize ($\Box_a$N)}
    \quad
    \begin{array}{c}
      \varphi
      \\\hline
       [Y]\varphi
    \end{array}
    \,\text{\footnotesize ($[Y]$N)}
    \quad
    \begin{array}{c}
      \varphi
      \\\hline
       [U,s]\varphi
    \end{array}
    \,\text{\footnotesize (UN)}
    \]
  \end{center}
  \caption{The theory $\YDEL$; formulas are all in $\lydel$}
  \label{figure:YDEL}
\end{figure}

\begin{theorem}
  \label{theorem:ydel-completeness}
  The theory of $\YDEL$ is sound and complete with respect to
  $\restricted_\pt$.
\end{theorem}
\begin{proof}
  Soundness for most of the axioms is straightforward. As such, we
  will only go through the proofs for the $[Y]$-reduction axiom and
  (UN).

  \begin{itemize}
    
  \item Soundness for $[U,s][Y]\varphi \iff (\pre^U(s)\imp\varphi)$.

    We first prove the left-to-right direction of the
    equivalence. Suppose $M,w \models [U,s][Y] \varphi$ and $M,w
    \models \pre^U(s)$. This implies that for every $(v,t) \leadsto^{M
      \oplus U} (w,s)$, we have $M \oplus U, (v,t) \models
    \varphi$. By definition of $M\oplus U$, we have $(w, \flat)
    \leadsto^{M \oplus U} (w,s)$, and so $M \oplus U (w,\flat) \models
    \varphi$. By Lemma~\ref{lemma:flat-bisim}, $M,w \models \varphi$.

    Now we prove the right-to-left direction.  Suppose $M,w \models
    \pre^U(s)$ and $M,w \models \varphi$ (the case where
    $M,w\not\models \pre^U(s)$ is immediate).  By definition of
    $M\oplus U$ and uniqueness of the past, $(w,t) \leadsto^{M \oplus
      U} (w,s)$ implies $t=\flat$. By the fact that $M,w \models
    \varphi$ and Lemma~\ref{lemma:flat-bisim}, we have $M \oplus U,
    (w,\flat) \models \varphi$. Hence $M,w \models [U,s][Y] \varphi$.

  \item Soundness for (UN) follows from the fact that restricted
    models are closed under the operation $M\mapsto M\oplus U$
    (Corollary~\ref{corollary:ydel-closure}).

  \end{itemize}
  Completeness follows a similar argument as was used for $\RDETL$
  (Theorem~\ref{theorem:rdetl-completeness}).
\end{proof}

\begin{corollary}
  For each $\varphi\in\ldetl$, we have:
  \[
  \vdash_\YDEL\varphi \qquad\text{iff}\qquad
  \vdash_\RDETL\varphi^\sharp\enspace.
  \]
\end{corollary}
\begin{proof}
  By Theorems~\ref{theorem:ydel-completeness},
  \ref{theorem:ydel-rdetl}, and
  \ref{theorem:rdetl-completeness}.
\end{proof}

\section{Conclusion}

We have presented \emph{Dynamic Epistemic Temporal Logic} ($\DETL$), a
general framework for reasoning about transformations on Kripke models
with a designated timekeeping relation $\leadsfrom$.  Our ``temporal''
action models are a generalization of the atemporal action models
familiar from Dynamic Epistemic Logic.  We showed by way of a number
of examples how temporal action models can be used to reason about and
control the flow of time.  We also highlighted some key design choices
that allow this framework to avoid conceptual complications relating
to time and that enable us to define actions that preserve a complete
copy of the past state of affairs.  This leads to one natural choice
for understanding time: the time of a world is the depth of that world
(i.e., the maximum number of ``backward'' temporal steps one can take
from that world, whenever this maximum exists).

Kripke models with a designated timekeeping relation are essentially
the models of Epistemic Temporal Logic.  Therefore, one way of
understanding our work is as follows: we extend the domain of action
model operations from those on (atemporal) Kripke models to (what are
essentially) the models of Epistemic Temporal Logic.  
We showed that a number of properties that may arise in the latter 
models---such as Persistence of Facts, Perfect Recall, and Synchronicity---are
preserved under the application of temporal action models that themselves satisfy certain
related properties.  
This makes it possible to use our $\DETL$
framework to develop Dynamic Epistemic Logic-style theories of
temporal Kripke models. Such logics can be used to reason about
objective changes in time along with the agents' basic and
higher-order knowledge and beliefs about changes in this structure.
As an example, we showed how the $\DETL$ approach can be used to
define the logic $\RDETL$ of ``restricted'' Dynamic Epistemic Temporal
Logic, which is essentially the Dynamic Epistemic Logic of synchronous
actions with the ``yesterday'' temporal operator $[Y]$.  We proved
that $\RDETL$ reasoning captures the reasoning of $\YDEL$, the first
Dynamic Epistemic Logic of synchronous time with the yesterday modal.

The $\DETL$ approach is not, however, limited to synchronous systems.
We presented one example where a synchronous model is transformed into
an asynchronous one, leaving one agent sure that two clock ticks
occurred, and the other uncertain as to whether it was one or two.  We
contrasted this with a synchronous variant in which the agents'
knowledge about atemporal information is the same, but the knowledge
change is compressed into a single clock tick that is common knowledge.  Here
we see that the difference is easily discernable by a simple
examination of the temporal action models involved.  In essence, our
theory extends the types of knowledge change describable by
atemporal action models to the temporal setting, which gives us a
great deal of control as to the relationship between how much time
passes and what the agents perceive of this passage.  And we of
course also inherit many features (and drawbacks) of the atemporal
action model approach.

In closing, we mention a recent study of time in Dynamic Epistemic
Logic that looks at asynchronous systems \cite{deglowwit11}.  The
basic idea is that an atemporal action model operates on a temporal
Kripke model in such a way that an agent experiences a single clock
tick if her knowledge of atemporal information changes, but
will otherwise be uncertain as to whether the clock ticked.  So, for
example, if agent $a$ does not know $p$, a public announcement of $p$
will transform her knowledge in a synchronous manner: the clock will
tick once, she will learn $p$, and she will know that the clock ticked
once.  But if the public announcement of $p$ then occurs again, the
clock will tick but, since she already knows $p$ and hence her
atemporal knowledge will not change, she will be uncertain as to
whether the clock ticked once or not at all.  The result is an
asynchronous situation.

Though we have shown (by way of an example) that $\DETL$ can reason
about some asynchronous updates, we have not proved that it can reason
about every such update.  Nor have we shown that it can reason about a
certain class of asynchronous updates that can be independently
identified according to some desirable properties it satisfies.  In
particular, it is not clear if there is a $\DETL$ action model for all the
updates that can be produced by the framework of \cite{deglowwit11}.  
Moreover, the latter approach is based on ``protocols'' constraining the sequences of actions
that can occur, something we have left out of the present study for
simplicity.  Another complication is that the asynchronous updates of
\cite{deglowwit11} essentially insert agent arrows $\to_a$ based on
whether a certain knowledge condition is satisfied, whereas our
temporal action models do not allow us to conditionally insert arrows.
This suggests that there may be connections with ``arrow update
logics'' \cite{KooRen11:RSL,KooRen11:TARK} that allow such conditions
on arrows.  In particular, it has been shown that generalized
arrow updates are equivalent to atemporal action models in terms of
update expressivity \cite{KooRen11:TARK}.  Therefore, an arrow update
version of $\DETL$ might suggest a natural way to represent
asynchronous updates like those of \cite{deglowwit11}, and this may
turn out to be equivalent in update expressivity to our present
approach, just as in the atemporal case.  If this is so, then it may
be the case that ``conditional'' arrow changes are already within the
scope of our current approach, albeit indirectly.

Regardless, we believe that $\DETL$ presents a viable option for
developing Dynamic Epistemic Logic-style theories of Epistemic
Temporal Logic.  While this paper (and its early predecessor
\cite{RenSacYap09}) present the first steps of this study, there is
clearly still much more work to be done.


\bibliographystyle{plain}
\bibliography{rsy-detl}

\end{document}